\documentclass{article}
\pdfoutput=1
\usepackage{graphicx} 

\usepackage{jheppub}
\usepackage{hyperref}
\usepackage{amsmath}
\usepackage{amsfonts}
\usepackage{graphicx}
\usepackage{caption}
\usepackage{subcaption}
\usepackage{placeins}
\usepackage{tikz}
\usepackage{tikz-feynman}
\usepackage{feynmp-auto}

\usepackage[english]{babel}
\usepackage{amsthm}
\newtheorem{theorem}{Theorem}
\newtheorem{lemma}{Lemma}
\newtheorem{example}{Example}
\newtheorem{definition}{Definition}
\theoremstyle{remark}
\newtheorem{remark*}{Remark}

\usepackage{array}
\usepackage{mathtools}
\usepackage{amsmath}
\usepackage{amsfonts}

\title{
The Tropical Geometry of Subtraction Schemes
} 

\author[a]{Giulio Salvatori}
\affiliation[a]{Max-Plank-Instit\"ut fur Physik, Werner-Heisenberg-Institut, D–85748 Garching bei M\"unchen, Germany}

\emailAdd{giulios@mpp.mpg.de}

\date{\today}

\abstract{
We study the construction of local subtraction schemes through the lenses of tropical geometry.
We focus on individual Feynman integrals in parametric presentation, and think of them as particular instances of \emph{Euler integrals}.
We provide a necessary and sufficient condition for a combination of Euler integrands to be locally finite, i.e. to be expandable as a Taylor series in the exponent variables directly under sign of integration.
We use this to construct a local subtraction scheme that is applicable to a class of Euler integrals that satisfy a certain geometric property.
We apply this to compute the Laurent expansion in the dimensional regulator $\epsilon$ of various Feynman integrals involving both UV and IR singularities, as well as to generalizations of Feynman integrals that arise in effective field theories and in phase-space integrations, for which we provide new analytic results.
}

\begin{document}

\maketitle

\section{Introduction}

In high-energy physics, crucial quantities such as scattering amplitudes and cross-sections are typically represented as complicated integrals.
A common theme of many modern developments in the field is to try to derive the properties of the final integrated result directly from the analytical properties of \emph{the integrand}.
We will apply this philosophy to the computation of ultraviolet (UV) and infrared (IR) divergences of Feynman integrals, which in dimensional regularization appear as poles in the regulator $\epsilon$.
Traditionally this is done by constructing \emph{subtraction schemes}\footnote{In the literature the terminology of ``subtraction scheme'' is usually applied to the computation of actual observables, such as amplitudes and cross-sections, rather than for individual diagrams. Here we are adopting the graph-by-graph point of view on renormalization spelled out in \cite{Collins:1984xc}}, that is by rewriting 
\begin{align}
    I(\epsilon) = \left[I(\epsilon)-I^{\rm ct}(\epsilon)\right]+I^{\rm ct}(\epsilon),
    \label{eq:subtractionscheme}
\end{align} where the ``counterterm'' $I^{\rm ct}(\epsilon)$ is a simpler (i.e. already known) integral with the same pole part as the desired one $I(\epsilon)$, so that the ``renormalized'' expression $I^{\rm ren}(\epsilon)$ is finite as $\epsilon \to 0$.
We will be interested in promoting this understanding of UV and IR divergences to the integrand level.
The motivation is two-fold. From a computational point of view, one may be interested in the direct evaluation of the finite quantity $I^{\rm ren}(\epsilon)$ by integrating, numerically or analytically, an appropriate renormalized integrand; rather than computing separately the quantities $I(\epsilon)$ and $I^{\rm ct}(\epsilon)$ and explicitly exposing the cancellation of the poles in $\epsilon$. For this purpose, a result such as \eqref{eq:subtractionscheme} is not enough, as we will illustrate with simple examples below.
On a more conceptual level, one may hope that understanding subtraction schemes at the level of the integrand may unearth geometrical structures underpinning the physics of renormalization, in analogy to how positive geometries have given a new perspective on the origin of locality and unitarity of scattering amplitudes - and simplified their computation \cite{Arkani-Hamed:2013jha,Arkani-Hamed:2017tmz,Arkani_Hamed_2018,arkanihamed2023loop}.

It will be helpful to frame the problem in a larger mathematical context. We will consider integrals of the form
\begin{align}
    I({\bf s}; \nu, c) = \int_{\mathbb{R}^n_{\ge0}} \frac{d\alpha}{\alpha}\mathcal{I}({\bf s};\nu,c),
    \label{eq:euler}
\end{align}
where we have factored out the canonical logarithmic differential form,
\begin{align}
    \frac{d\alpha}{\alpha} \coloneqq \frac{d\alpha_1}{\alpha_1} \wedge \dots \wedge \frac{d\alpha_n}{\alpha_n},
\end{align}
from the rest of the integrand, which is of the form
\begin{align}
   \mathcal{I}({\bf s};\nu, c) = \prod_{i=1}^n \alpha_i^{\nu_i} \prod_{j=1}^{m} P_j({\bf s},\alpha)^{-c_j}.
    \label{eq:eulerintegrand}
\end{align}
 The quantities $P_j({\bf s},\alpha)$ are polynomials in the variables $\alpha = (\alpha_1, \dots, \alpha_n)$ with coefficients collectively denoted by ${\bf s}$.
Integrals of the type \eqref{eq:euler} are known as \emph{Euler integrals} \cite{berkesch2014euler,gelfand1990generalized,nilsson2013mellin}, and are ubiquitous in physics. String scattering amplitudes (and generalizations thereof \cite{Arkani-Hamed:2019mrd}), Feynman integrals in parametric form, integrals which appear in the Method of Regions \cite{smirnov1999problems} and certain phase-space integrals \cite{smirnov2024expansion}, all fall into this class. See \cite{Matsubara-Heo:2023ylc} for a physically-oriented introduction to Euler integrals.

An Euler integral defines a function $I({\bf s}; \nu, c)$ that depends on two sets of variables: the coefficients ${\bf s}$ and the exponents $(\nu,c)$ of the polynomials appearing in \eqref{eq:eulerintegrand}.
It turns out that an Euler integral is a meromorphic function of the exponents, while the functional dependence on the coefficients is considerably more complicated, usually involving multi-valued functions which generalize the logarithm.

In this paper, we shall be concerned only with the dependence on the exponents. More precisely, we will consider the case where all exponents depend affinely on a common parameter $\epsilon$, and we will be interested in expressing $I({\bf s};\epsilon)$ as a Laurent series around $\epsilon=0$, 
\begin{align}
    I({\bf s};\epsilon) = \tau_\epsilon I({\bf s};\epsilon) = \sum_{i=-N}^\infty \epsilon^i I^{(i)}({\bf s}).
\end{align}
From now on, we will use the symbol $\tau$ to denote an operator that returns the appropriate (i.e. Taylor, Laurent or Puiseux) series of the function to which it is applied.
Naively, one may think that the coefficients $I^{(i)}({\bf s})$ may be obtained by Taylor expanding the Euler integrand $\cal{I}({\bf s},\epsilon)$ with respect to the exponents, and then integrate each term of the resulting series.
\begin{align}
    \tau_\epsilon \int_{\mathbb{R}^n_{\ge0}} \frac{d\alpha}{\alpha}  \mathcal{I}({\bf s},\epsilon) \stackrel{?}{=} \int_{\mathbb{R}^n_{\ge0}} \frac{d\alpha}{\alpha} \tau_\epsilon \mathcal{I}({\bf s};\epsilon).
\end{align}
However, this is wrong since in general we expect $I(\epsilon)$ to have poles, which would not be generated by Taylor expanding the integrand in $\epsilon$.
Therefore, it is at the very least necessary to expose the pole structure of $I(\epsilon)$, before performing a series expansion at the integrand level.

One systematic method to do so was introduced by Nilsson and Passare \cite{nilsson2013mellin,berkesch2014euler}, and consists in applying a sequence of integration by parts identities that reduce $I({\bf s};\epsilon)$ to a combination of finite Euler integrals multiplied by poles in $\epsilon$. The expansion of the finite integrals can then be obtained by expanding the corresponding integrands directly under sign of integration.
This method was advocated in \cite{Panzer_2014} as a strategy to deal with divergent Feynman integrals, in view of its crucial advantage of preserving \emph{linear reducibility}, and was made more efficient in \cite{von_Manteuffel_2015}.
While this is a completely algorithmic approach to the problem, it has an unsatisfactory feature, which we now illustrate.
Let us first explain how poles in the $\epsilon$ arise considering a toy example,
\begin{align}
    \int_{0}^1 \int_{0}^1 \frac{dx}{x}\frac{dy}{y} x^\epsilon \mathcal{J}(x,y) ~ \sim \frac{1}{\epsilon} \int_0^1 \frac{dy}{y} \mathcal{J}(0,y) + \mathcal{O}(\epsilon).
\end{align}
As the regulator $\epsilon \to 0$, the integral diverges due to the naked logarithmic singularity at $x=0$. Therefore, it is safe to assume that the integral is dominated by the region of integration around $ x \sim 0$, where the integrand is well approximated by $\mathcal{J}(x,y) \sim \mathcal{J}(0,y)$.
This shows that the pole is proportional to the integral of $\cal{J}$ over $y$-space.
The example shows that \emph{each pole in $\epsilon$ allows to perform one integral exactly}. Therefore, it should be possible to express the coefficients $I^{(i)}({\bf s})$ as simpler integrals over lower dimensional spaces. This is a crucial property of Euler integrals that is unfortunately lost in the Nilsson-Passare approach.

Here we propose a different strategy, based on the development of a suitable \emph{local subtraction scheme}. We will explain how to build a collection of counterterms that reproduce the pole structure of an Euler integral, so that subtracting them from the original integrand yields a renormalized integrand that can be safely expanded under the sign of integration. Furthermore, the counterterm integrands can be integrated exactly, explicitly manifesting the pole structure of the Euler integral under consideration.
In other words we look for a counterterm integrand $\mathcal{I}^{\rm ct}$ such that,
\begin{align}
    \int_{\mathbb{R}^n_{\ge0}}  \frac{d\alpha}{\alpha} \left(\mathcal{I} - \mathcal{I}^{\rm ct}\right) = \mathcal{O}(\epsilon^0) {\quad \rm and\quad } \int_{\mathbb{R}^{n}_{\ge0}} \frac{d\alpha}{\alpha} \mathcal{I}^{\rm ct} = \frac{1}{\epsilon^m}\int_{\mathbb{R}^{n-m}_{\ge0}}\mathcal{J}
\end{align}
and then iterate the procedure on $\mathcal{J}$.

The heart of the matter is that just because the renormalized integrand yields a finite quantity, i.e. of order $\mathcal{O}(\epsilon^0)$, it is not necessarily possible to expand it in $\epsilon$ directly under sign of integration. A simple counter-example is given by
\begin{align}
    0 = \frac{1}{\epsilon}-\frac{1}{\epsilon} = \int_0^1 \frac{dx}{x} \left(x^\epsilon- 2 x^{2\epsilon}\right),
\end{align}
expanding the integrand on the RHS in powers of $\epsilon$ produces divergent integrals, and is therefore not allowed.
This motivates the introduction of the stronger notion of \emph{local finiteness}: we say that a combination of Euler integrands is locally finite if integration and Taylor expansion commute,
\begin{align}
    \tau_\epsilon \int \frac{d\alpha}{\alpha }\mathcal{I}^{\rm ren} = \int \frac{d\alpha}{\alpha}\tau_\epsilon \mathcal{I}^{\rm ren} .
\end{align}

The simple counterexample just discussed already points at the criterion that guarantees local finiteness. The source of divergences of an Euler integral are boundary regions in the domain of the integration\footnote{More precisely, this is the case if the coefficients ${\bf s}$ are non-negative.}. If the local behavior of the counterterm integrand around these regions cancels the pole coming from the canonical logarithmic measure, then local finiteness holds.
The first subtlety is that these potentially problematic regions may be revealed only after a suitable change of variables. 

Assuming that the problematic regions have been correctly identified, we can try to locally eliminate the associated divergences by subtracting the asymptotic behavior of the integrand in these regions, as captured by an appropriate power expansion.
As a simple example, consider
\begin{align}
    I(\epsilon)^{\rm ren} = \int_0^1 \frac{dx}{x} \left[x^\epsilon (1+x)^\epsilon - x^\epsilon\right] = \int_0^1 \frac{dx}{x} (\log(1+x) \epsilon + \dots) = \frac{1}{\epsilon}\frac{\zeta(2)}{2} + \dots,
\end{align}
which is locally finite. In higher dimensions a subtlety arises: the \emph{overcounting of divergences}. 
Imagine that we tried the same strategy in the following example,
\begin{align}
    I(\epsilon)^{\rm ren} = \int_{[0,1]^2} \frac{dx}{x}\frac{dy}{y} \left[x^\epsilon y^\epsilon (x + y + x y)^\epsilon - x^\epsilon y^\epsilon (y)^\epsilon - x^\epsilon y^\epsilon (x^\epsilon) \right].
    \label{eq:counterexample2}
\end{align}
Where we have subtracted from an integrand its local expressions around $x=0$ and around $y=0$.
However, \eqref{eq:counterexample2} is \emph{not} locally finite. The issue is that we are subtracting twice the singularity at the origin.
We could try to fix the problem by adding back this singular behaviour, but this brings us to a puzzle: what is the series expansion of $(x+y+x y)^\epsilon$ around the origin?
We can formalize this by defining the operator $\tau_z$  that acts on an integrand by stripping the factor $z^\epsilon$, then performs the Taylor expansion in $z$ (up to some appropriate order), and finally multiplies back the factor $z^\epsilon$. The combination appearing in \eqref{eq:counterexample2} can then be written as
\begin{align}
    I(\epsilon)^{\rm ren} = \int_{[0,1]^2} \frac{dx}{x}\frac{dy}{y} (1-\tau_x-\tau_y) \mathcal{I}.
\end{align}
The issue is that the operators $\tau_x$ and $\tau_y$ \emph{do not commute} when acting on the above integrand,
\begin{align}
    \tau_x \tau_y \mathcal{I} \ne \tau_y \tau_x \mathcal{I} ,
\end{align}
- which is ultimately due to the fact that $(x+y+x y)$ vanishes at the origin - and therefore we cannot add  neither $\tau_x  \tau_y \mathcal{I}$ nor $\tau_y  \tau_x \mathcal{I}$ to fix the overcounting.
The algorithm of \emph{Sector Decomposition} \cite{Binoth:2000ps, Binoth:2003ak,Bogner:2007cr, Kaneko:2010kj,kaneko2010geometric} solves this problem by breaking the integral into ``sectors'', parametrized in such a way that the integrand is no longer vanishing at the origin. In the above example, this is achieved by rescaling $y \to x y$ and separating the domain of integration into $x \in [0,1]$ and $x \in [1,\infty]$.
The disadvantage of this procedure is that spurious structures are inevitably introduced. This is reflected, for instance, in the appearance of intermediate transcendental numbers that cancel in the sum over the sectors. Clearly, this is also accompanied by an unnecessary computational complexity.

Here we follow a different approach. Having identified the regions responsible for divergences, we will engineer counterterms that on the one hand correct the local behavior around these regions, and on the other are globally defined on the whole domain of integration; thus avoiding the artificial decomposition in sectors.
Obviously, the counterterm introduced to fix a specific region should be easier than the original integrand, so that the integration in the problematic region can be performed exactly, exposing the corresponding pole.
The simplest way to produce such a counterterm would be to use the local behavior of the original integrand itself. However, this would not be globally defined, and would introduce new divergent regions far away from the one under consideration. The challenge is then to modify this local expression so that no new divergences are introduced, while maintaining enough control over it so that is still possible to analytically perform the integration that expose the poles. 

We will address all of the above issues by using rudiments of \emph{Tropical Geometry} \cite{maclagan2015introduction} to study the local behavior of Euler integrands and the commutativity of series expansions.
In particular, the notion of Newton polytope will play a predominant role. 
Its importance in the study of Euler integrals is well known \cite{nilsson2013mellin,gelfand1990generalized, Arkani-Hamed:2022cqe} and it has played an important role in the development of algorithms to compute Feynman integrals, including sector decomposition.

Our first major result is a necessary and sufficient condition for a combination of Euler integrands to be locally finite, Th. \ref{th:localfinite3}.
Building on this, we develop a simple-minded subtraction scheme, inspired by ideas discussed in \cite{Brown:2019wna} in the context of string amplitudes, and already considered in the context of UV divergent Feynman integrals \cite{Hillman:2023ezp}.
These earlier works focus on specific instances of divergent Euler integrals: tree-level string amplitudes and UV divergent Feynman integrals, respectively.
On the other hand, our subtraction scheme is applicable to any Euler integral that satisfies a certain geometrical property.

This paper is structured as follows.
We begin by offering a short review of necessary notions from polyhedral geometry, in Section \ref{sec:polygeom}.
In Section \ref{sec:newt} we introduce the Newton polytope and explain its importance in studying the behavior of a polynomial under a rescaling of its variables.
In Section \ref{sec:euler} we move to Euler integrals and relate the properties of the corresponding integrands to the geometric data of the Newton polytopes.
We also review some important algorithms that require knowledge of the polytope, the \emph{Nilsson-Passare} analytical continuation and \emph{Sector Decomposition}.
Building on these preliminaries, in Section \ref{sec:trop} we present the first nontrivial result of this paper, a theorem for local finiteness of Euler integrals.
In Section \ref{sec:scheme} we apply this to construct a simple-minded subtraction scheme. We show its usefulness in a series of examples in \ref{sec:examples}.
Finally, we discuss directions for further exploration.

\section{Polyhedral Geometry}
\label{sec:polygeom}

In this section, we offer a summary of the basic facts of polyhedral geometry that are relevant for this paper. Proofs can be found in standard textbooks such as \cite{ziegler1995lectures}.

We denote the standard coordinates on $\mathbb{R}^n$ by ${\bf z} = (z_1, \dots, z_n)$.
Given a collection of vectors $\{\rho_i\}_{i \in I}$ in $\mathbb{R}^n$ we define their span and positive span as
\begin{align}
    \mathrm{Span }_K \{\rho_i\}_{i \in I} \coloneqq \{{\bf z} | {\bf z} = \sum_{i \in I} \mu_i \rho_i, \mu_i \in K \},
\end{align}
for $K$ in $\mathbb{R}$ and $\mathbb{R}_{\ge0}$, respectively. We will also write $\mathrm{Span}$ and $\mathrm{Span}_+$ to ease the notation.
Given two sets $X,Y \subset \mathbb{R}^n$, we define their \emph{Minkowski sum} to be the set
\begin{align}
    X + Y = \{ v \in \mathbb{R}^n | v = x + y, x \in X, y \in Y\}.
\end{align}
Given two matrices $A,B$ with the same number of rows, we denote by $A|B$ the matrix formed by listing the columns of $A$ followed by those of $B$.

A \emph{cone} $\sigma$ in $\mathbb{R}^n$ is any set closed under taking positive spans. That is, if vectors $\{\rho_i \}_{i \in I}$ are in $\sigma$, then so is $\rm{Span }_+(\{\rho_i\}_{i \in I})$.
A cone can be written as the positive span of a minimal collection of vectors called \emph{rays}, which themselves cannot be written as the positive span of other elements of the cone. In other words, the rays $\mathrm{Rays\ } \sigma$ of a cone $\sigma$ form the minimal set such that
\begin{align}
    \sigma = \mathrm{Span}_{+}\ \mathrm{Rays\ } \sigma.
\end{align}
The interior of $\sigma$ is the strictly positive span of its rays, that is
\begin{align}
    \sigma^+ \coloneqq \{{\bf z} | {\bf z} = \sum_{\rho \in \mathrm{Rays\ }\sigma} \mu_\rho \rho, \mu_\rho > 0 \}.
\end{align}
The dimension of a cone is the largest dimension of a ball properly contained in the cone.
A cone is said to be \emph{pointed} if it does not contain any proper affine space, and \emph{simplicial}  if its rays are linearly independent.
If a cone is not simplicial then the number of rays is greater than the dimension of the cone.

A \emph{fan} $\Sigma$ is a collection of cones closed under intersections. That is, if $\sigma$ and $\sigma'$ are in $\Sigma$ so is $\sigma \cap \sigma'$. We denote by $\Sigma(m)$ the subset of $\Sigma$ formed by the cones of dimension $m$.
A fan $\Sigma'$ is said to \emph{refine} $\Sigma$ if for all $\sigma' \in \Sigma'$ there is a $\sigma \in \Sigma$ such that $\sigma' \subset \sigma$.
Given two fans $\Sigma$ and $\Sigma'$ one can produce a new fan that refines both,
\begin{align}
    \Sigma | \Sigma' \coloneqq \{\sigma \cap \sigma', \sigma \in \Sigma, \sigma' \in \Sigma'\},
\end{align}
and that is therefore called the \emph{common refinement} of $\Sigma$ and $\Sigma'$.
We say that a fan is pointed (or simplicial) if all of its cones are.
We can endow a fan with a lattice structure by declaring two cones to be compatible if there is a cone of the fan containing both.
More generally, we say that two sets $X,Y \subset \mathbb{R}^n$ are compatible with respect to $\Sigma$ if there is a cone $\sigma \in \Sigma$ containing both.
We denote by $\Sigma|_X$ the set of cones of $\Sigma$ compatible with $X$. 

A \emph{polytope} in the V-presentation is a subset $\mathcal{P} \subset \mathbb{R}^n$ given by the convex hull of a collection of points which we list as columns of the matrix $\mathcal{Z} = \left({\bf z}_1 | \dots | {\bf z}_v\right) \subset \mathbb{R}^n$, that is,
\begin{align}
    \mathcal{P} = \mathrm{Conv\ } \mathcal{Z} \coloneqq
    \left\{{\bf z} | {\bf z}  = \sum_{i=1}^v \mu_i {\bf z_i}, \quad \mu_i \ge 0, \quad \sum_{i=1}^v \mu_i = 1 \right\}
\end{align}
We say that the presentation is non-redundant if for any subset $\mathcal{Z}' \subset \mathcal{Z}$ the polytope $\mathrm{Conv\ } \mathcal{Z}'$ is strictly contained into $\mathrm{Conv\ } \mathcal{Z}$, in which case we say that the columns of $\mathcal{Z}$ are the \emph{vertices} of $\mathcal{P}$. The fundamental theorem of polyhedral geometry is that a polytope admits an alternative dual description in terms of half-spaces, or H-presentation.
A polytope in H-presentation is defined as the locus where a set of finitely many affine equations and affine inequalities are satisfied.
\begin{align}
    \mathcal{P} = \bigcap_{i=1,\dots,f} \{d_i - \rho_i \cdot z_i \ge 0\} \bigcap_{j=1,\dots,e} \{d_j - \rho_j \cdot z_j = 0\}.
\end{align}
We can write this more compactly by listing the constants $d_i$ (resp. $d_j$) and vectors $\rho_i$ (resp. $d_j$) as rows of matrices $d_f$ (resp. $d_e$) and $\rho_f$ (resp. $\rho_e$) and writing
\begin{align}
    \mathcal{P} = \rm{Poly}(d_f|\rho_f, d_e|\rho_e),
\end{align}
Note that the presentation is not unique, since we can always add linear combinations of the equations to the inequalities.
We say that the presentation is non-redundant if $\mathcal{P}$ is strictly contained in the polytope obtained by deleting any proper subset of the rows of $(d_f|\rho_f)$ and/or of the rows $(d_e|\rho_e)$.
In this case we refer to the inequalities encoded by the rows of $(d_f|\rho_f)$ as \emph{facets}.
From now on we will always assume of working with a non-redundant presentation.
Note that the notion of vertices and facets are dual to each other.

The dimension of a polytope is the largest dimension of a ball contained in the polytope; this can be strictly smaller than the dimension of the ambient projective space where the polytope is defined. In H-presentation, the dimension of a polytope is given by the corank of the linear system ofequations.

Let $\phi : \mathbb{R}^n \to \mathbb{R}^n$ be an invertible linear map, represented by the matrix $M$, and $\mathcal{P} \subset \mathbb{R}^n$ be a polytope with presentations
\begin{align}
    \mathcal{P} = \mathrm{Conv}(\mathcal{Z}) = \mathrm{Poly}(d_f|\rho_f, d_e|\rho_e).
\end{align}
Then $\phi(\mathcal{P}) \cdot \mathbb{R}^m$ is also a polytope, presented as
\begin{align}
    \phi(\mathcal{P}) = \mathrm{Conv}(M \cdot \mathcal{Z}) = \mathrm{Poly}(d_f|\rho_f \cdot M^{-1}, d_e|\rho_e \cdot M^{-1}).
\end{align}

A \emph{face} of a polytope is the locus where a finite number of inequalities of its H-presentation are saturated, i.e. they are replaced by equations.
With a slight abuse of terminology, we can say that a facet $\rho$ is face of co-dimension one, corresponding to saturating exactly one inequality $d_\rho -\rho \cdot {\bf z} \ge 0$. 
By definition, a face is also a polytope automatically given in a (possibly redundant) H-presentation. Saturating a collection of inequalities $\{d_i - \rho_i \cdot {\bf z}\}_{i \in I}$ may also imply that additional inequalities may be automatically saturated, in which case the co-dimension of the face may be higher than expected, see example \ref{ex:pyramid},
We say that two facets meet transversely if the face obtained by saturating both is of co-dimension two.
We say that a collection of $m$ facets meet transversely if all pairs of facets do. In this case the face $F$ obtained saturating all inequalities is of co-dimension $m$.
\begin{example}
    The pyramid shown in Fig. \ref{fig:segmentgeometry}(c) has H-presentation \begin{align*}
        \mathrm{Poly}\left(\left(\begin{array}{c}
          0 \\
          0 \\
          4 \\
          4 \\
          0 \\
    \end{array}\right), \left(\begin{array}{ccc}
         -2 & 0 & 1 \\
         0 & -2 & 1 \\
         2 & 0 & 1 \\
         0 & 2 & 1 \\
         0 & 0 & 1 \\
    \end{array}\right); (0), (0) \right),
    \end{align*}
    we have that $(d_1 - \rho_1 \cdot {\bf z})+(d_3 - \rho_3 \cdot {\bf z})=(d_2 - \rho_2 \cdot {\bf z})+(d_4 - \rho_4 \cdot {\bf z})$, which implies that on the locus where two opposite sides of the pyramid meet, also the other two sides do. Therefore neither of the two pair of facets meet transversely.
    The four facets all meet at the common face corresponding to the apex of the pyramid.
    \label{ex:pyramid}
\end{example}

Let $\rm{Faces\ } \mathcal{P}$ be the collection of faces of a polytope. We can endow it with a lattice structure by defining two faces to be compatible if and only if their intersection is a nonempty face of the polytope.
The face lattice of a polytope can be easily constructed from the knowledge of both of its presentations.
A consequence of the fundamental theorem on polytopes is that any linear functional attains its extremal values over a polytope on one of its faces. We use this to define a fan in the dual vector space $(\mathbb{R}^n)^\star$ of linear functionals acting on $\mathbb{R}^n$. It is called the \emph{normal fan} of  $\mathcal{P}$, and denoted by $\mathrm{Fan\ } \mathcal{P}$. 
First we assign a cone to each face via
\begin{align}
    \sigma_F = \{\rho \in \mathbb{R}^n\ |\ \rho \cdot {\bf z}\ \mathrm{is\  maximized\ at\ } F \},
    \label{eq:coneface}
\end{align}
For instance, if $F$ is the apex of pyramid of example \ref{ex:pyramid}, the corresponding cone is $\sigma_F = \mathrm{Span}_+ \{\rho_i\}_{i=1}^4$.
The normal fan of $\mathcal{P}$ is the collection of all these cones,
\begin{align}
    \mathrm{Fan\ }\mathcal{P} = \{\sigma_F, F \in \mathrm{Faces\ } \mathcal{P}\}.
\end{align}
Any fan $\Sigma \subset \mathbb{R}^n$ that arises as the normal fan of a polytope must be \emph{complete}, that is $\bigcup_{\sigma \in \Sigma} \sigma = \mathbb{R}^n$.
The assignment $F \to \sigma_F$ is a dimension-reversing isomorphism of lattices.
This means that two cones $\sigma_F$ and $\sigma_{F'}$ are compatible if and only if $F$ and $F'$ are.
If a collection of facets $\{\rho_i\}_{i \in I}$ of a non-redundant H-presentation of a polytope meet transversely, then their positive span $\mathrm{Span\ }_+ \{\rho_i\}_{i=1}^m$ of vectors in $\mathbb{R}^n$ is a cone of $\mathrm{Fan\ }\mathcal{P}$.
Under the linear map $\phi : \mathbb{R}^n \to \mathbb{R}^n$ the rays $\rho \in \Sigma(1)$, thought of as row vectors, transform to $\rho \cdot M^{-1}$. The lattice structure of the fan does not change.

The normal fan of a polytope is readily obtained from the data of its H-presentation and the face lattice. 
Suppose that there are no equations in the presentation of the polytope. Then the fan is pointed and its one-dimensional cones $\Sigma(1)$ are the positive spans of the vectors $\rho_i$ appearing in the H-presentation. The higher dimensional cones are given by positive spans of collections of rays pairwise compatible according to face lattice.
Now consider the case where the H-presentation contains equations $\{d_j - \rho_j \cdot {\bf}=0\}_{j\in J}$ defining an affine plane ${\bf V} \subset{\mathbb{R}^n}$.
Let $\pi_V$ be a projection on ${\bf V}$, that is, any surjective affine map $\pi_{\bf V} : \mathbb{R}^n \to {\bf V}$. 
The image $\pi_{\bf V}(\mathcal{P})$ is a polytope, with no equations in its H-presentation (with respect to some choice of coordinates on ${\bf V}$). Let $\Sigma'$ be its normal fan.
Consider the dual map $\pi_V^\star : {\bf V}^\star \to (\mathbb{R}^n)^\star$  defined by $\pi^\star_{\bf V}(\eta) = \eta \circ \pi_{\bf V}$. Then the cones of $\mathrm{Fan\ } \mathcal{P}$ can be described as $\pi^\star_{\bf V}(\sigma') + \mathrm{Span} \{\rho_j \}_{j \in I}$, with $\sigma' \in \Sigma'$.

Lower dimensional polytopes play an important role in the construction of the counterterms of the subtraction scheme presented in Section \ref{sec:scheme}.
Therefore, it is worth to illustrate the above definitions with a simple example.
\begin{example}
    Consider the polytope $\mathcal{P} = \mathrm{Conv\ }\{(1,0),(0,1)\}$.
    Its H-presentation is
    \begin{align}
        \mathcal{P} = \{1-\rho_+ \cdot {\bf z} \ge 0 \} \cap \{1- \rho_{-} \cdot {\bf z} \ge 0\} \cap \{1-\rho_e \cdot {\bf z} =0\},
    \end{align}
    with $\rho_{\pm}=(\pm 1, \mp 1)$, and $\rho_e = (1,1)$. Its normal fan $\Sigma = \mathrm{Fan\ } \mathcal{P}$ contains two two-dimensional cones, $\sigma_{\pm} = \{t_1\pm t_2 \ge 0\} = \mathrm{Span}_{+} \{\rho_{\pm}\}+ \mathrm{Span} \rho_e$.
    Since an equation is present in the H-presentation of $\mathcal{P}$, $\Sigma$ is not pointed.
    The polytope $\mathcal{P}$ arises as the Newton polytope of the polynomial $P=x+y$, a notion that will be introduced in the next section.  
\end{example}
\begin{figure}[h!]
    \centering
    \begin{subfigure}[b]{.3\textwidth}
        \centering        
        \includegraphics[width=\textwidth]{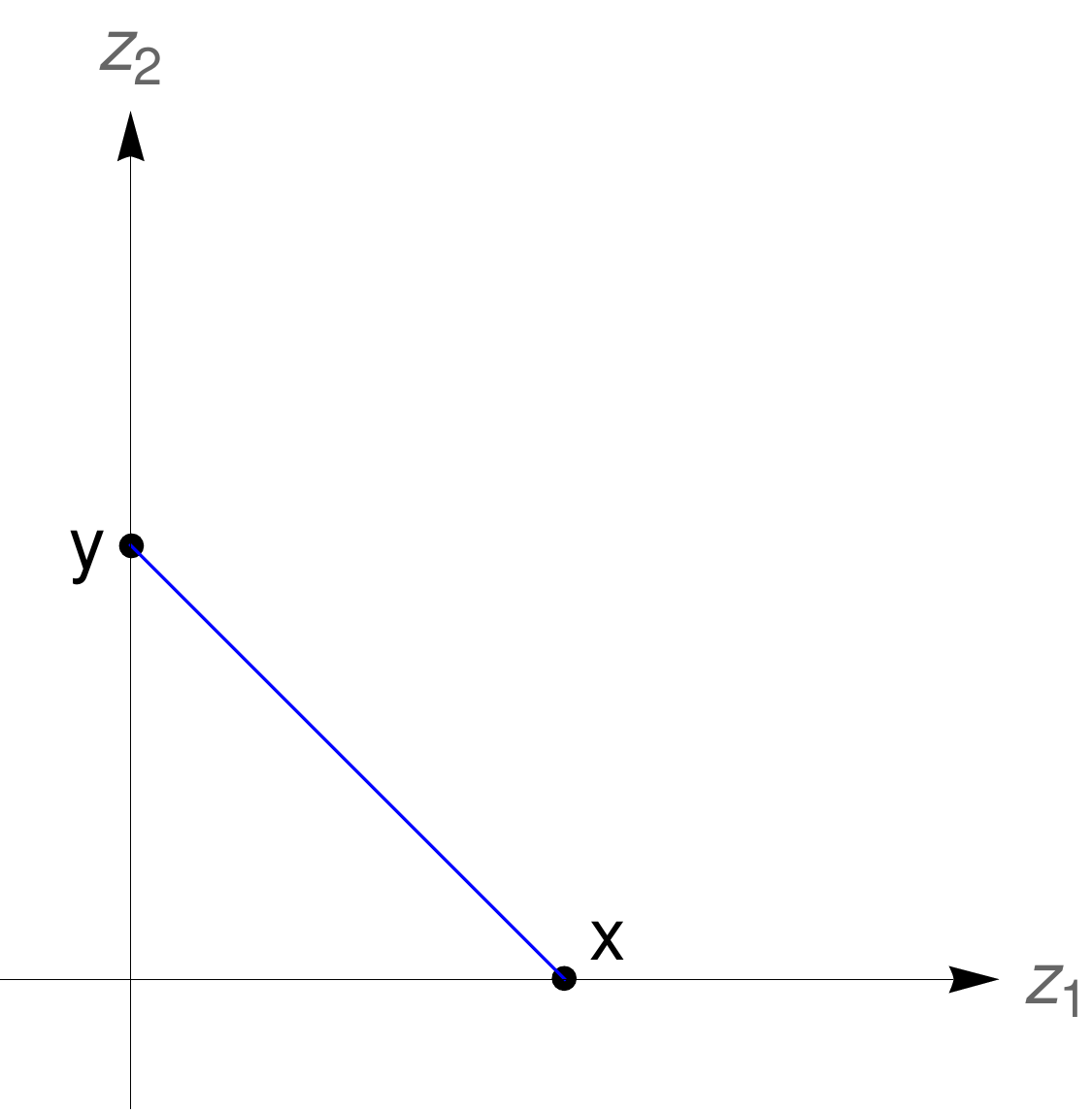}
        \caption{}
    \end{subfigure}%
    \hfill
    \begin{subfigure}{0.3\textwidth}
        \centering
        \includegraphics[width=\textwidth]{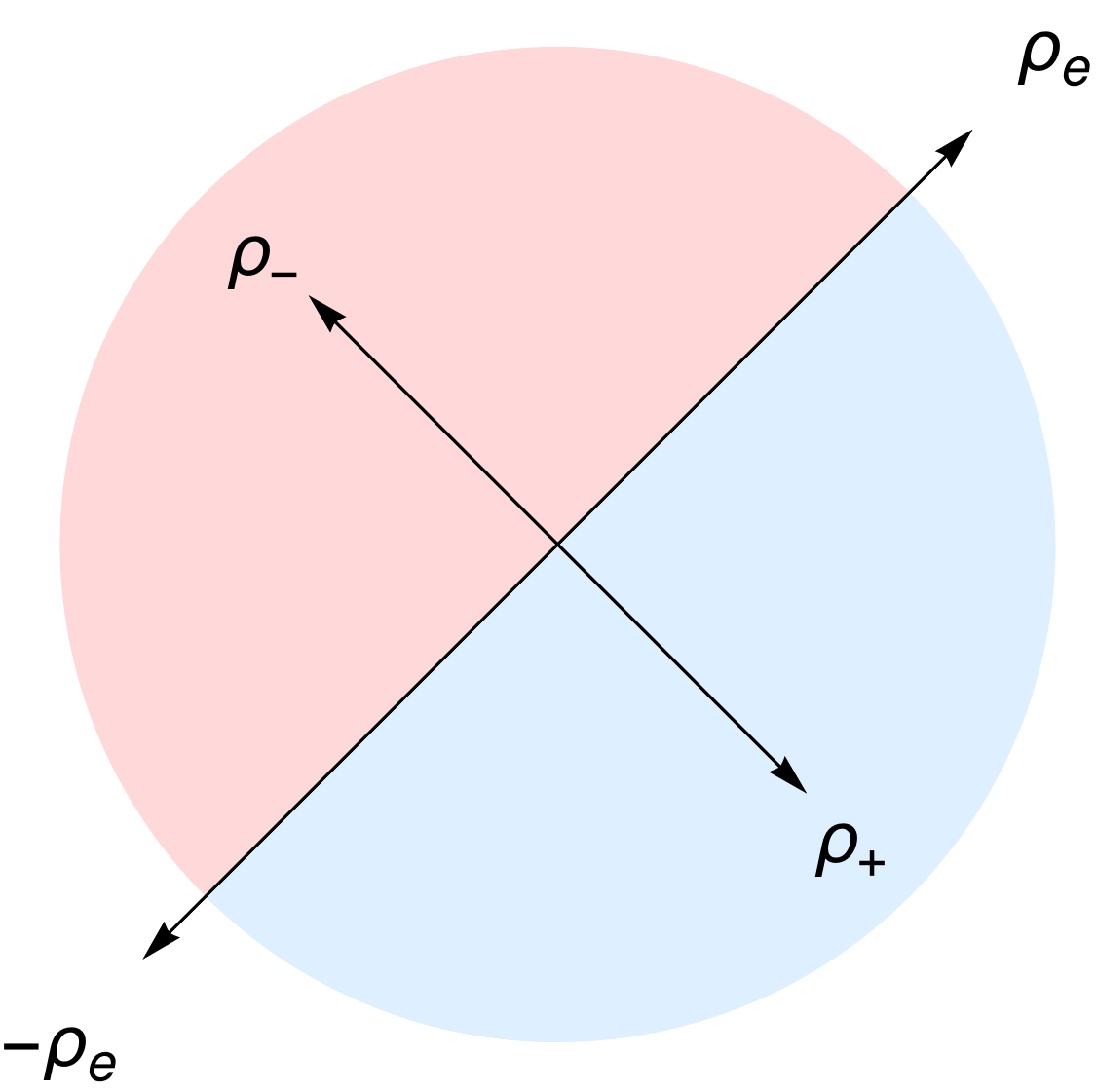}
        \caption{}
    \end{subfigure}
    \hfill
    \begin{subfigure}{0.3\textwidth}
        \centering
        \includegraphics[width=\textwidth]{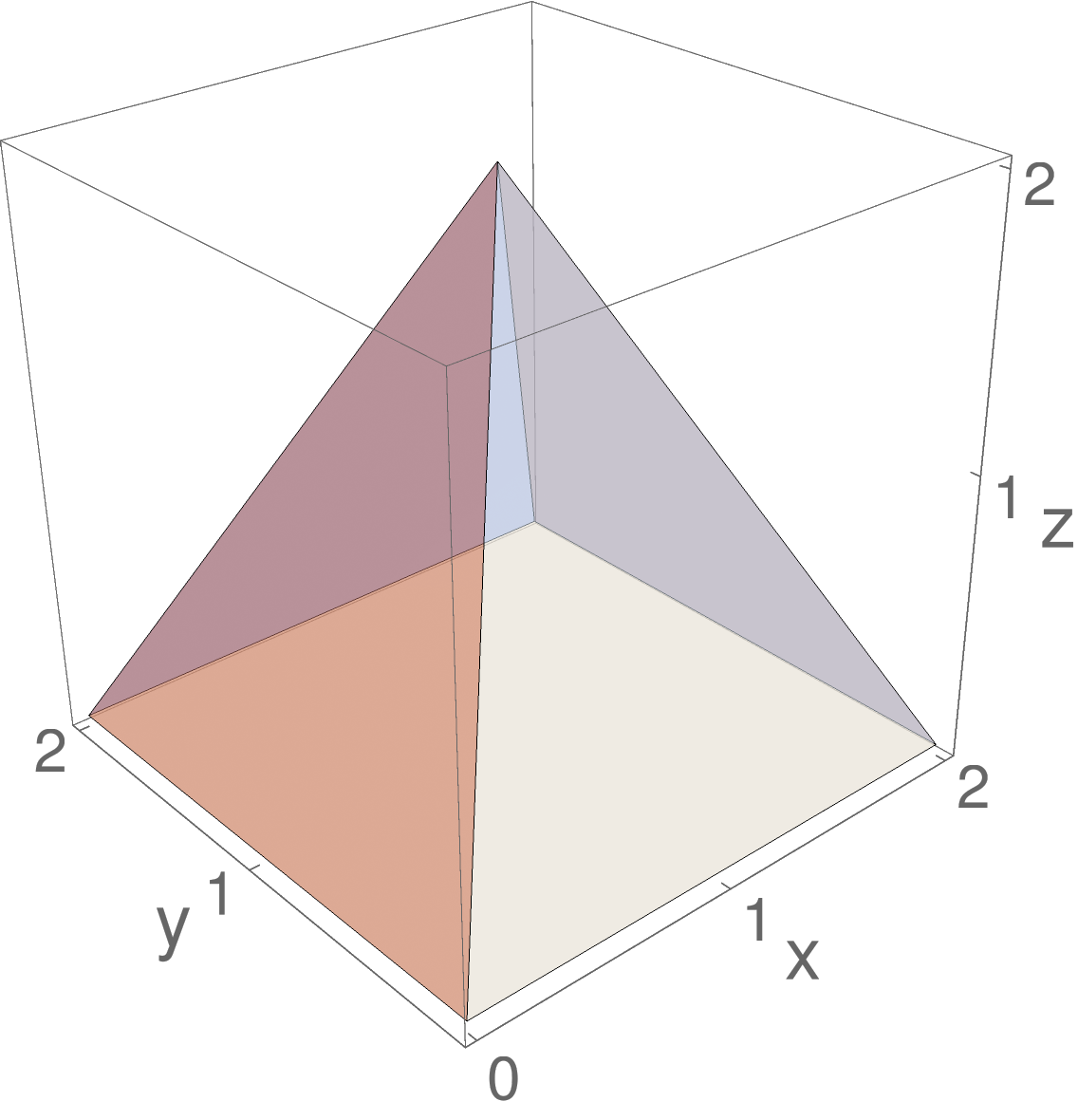}
        \caption{}
    \end{subfigure}
    \caption{The Newton polytope of $P=x+y$ (a) and its normal fan (b). The opposite facets of the pyramid (c) do not meet transversely}
    \label{fig:segmentgeometry}
\end{figure}

\section{Newton Polytopes}
\label{sec:newt}

Consider a Laurent polynomial in the variables $\alpha = (\alpha_1, \dots, \alpha_n)$,
\begin{align}
    P(\alpha) = \sum_{{\bf m} \in \mathbb{Z}^n} s_{\bf m} \alpha^{\bf m}, \quad
    \alpha^{\bf m} \coloneqq \alpha_1^{m_1} \dots \alpha_n^{m_n}.
\end{align}
The \emph{Newton polytope} of $P$ is defined as
\begin{align}
    \mathrm{Newt\ } P \coloneqq \mathrm{Conv\ } \{{\bf m} | {\bf m} \in \mathbb{Z}^n, s_{\bf m} \ne 0  \}.
\end{align}
In other words, the polytope $\mathrm{Newt\ } P$ is the convex hull of the points that appear as exponent vectors of the monomials of $P$.

The Newton polytope plays a central role in the study of Euler integrals, ultimately because its H-presentation contains the data necessary to understand the behavior of the polynomial $P(\alpha)$ under rescaling of the variables $\alpha$.
For any vector $\rho \in \mathbb{R}^n$, consider the rescaling
\begin{align}
    \alpha_i \to \alpha_i \lambda^{-\rho_i}, \quad i = 1, \dots, n,
    \label{eq:rescaling}
\end{align}
which we will also write more compactly as $\alpha \to \alpha \lambda^{-\rho}$.
Plugging \eqref{eq:rescaling} into $P$, and clearing all the denominators in $\lambda$, we get
\begin{align}
    P(\alpha \lambda^{-\rho}) = \lambda^{-T} \tilde{P}(\alpha,\lambda),
\end{align}
for some exponent $T \in \mathbb{R}$, where $\tilde{P}(\alpha,\lambda)$ is a \emph{polynomial} (not a Laurent polynomial) in positive real powers of $\lambda$. 
Clearly, $T$ is equal to the maximum of the scalar products between $\rho$ and the exponent vectors of the monomials of $P$. 
Anticipating the notation of Section \ref{sec:trop}, let us define the piecewise linear function 
\begin{align}
    \mathrm{Trop}\ P : \rho \mapsto \mathrm{Trop} P(\rho) \coloneqq \max_{{\bf m} \in \mathrm{Newt} P}  \rho \cdot {\bf m},
    \label{eq:tropPoly}
\end{align} so that $T = \mathrm{Trop\ } P(\rho)$. 
The polynomial $\tilde{P}$ is given by
\begin{align}
    \tilde{P}(\alpha,\lambda) = \sum_{{\bf m}\in \mathrm{Newt\ }P} s_{\bf m} \alpha^{\bf m} \lambda^{ \mathrm{Trop\ }P (\rho) - \rho \cdot {\bf m}} = P|_\rho(\alpha) + \mathcal{O}(\lambda),
\end{align}
in the second equation we have introduced the \emph{initial form} of $P$ along $\rho$, given by
\begin{align}
    P|_\rho(\alpha) = \sum_{{\bf m} | \rho \cdot {\bf m}  = \rm{Trop\ } P(\rho)} s_{\bf m} \alpha^{\bf m}.
\end{align}
In other words, the initial form is obtained from $P$ by keeping only the monomials whose exponent vector lies on the face $F \in \mathrm{Face\ Newt\ } P$ where the functional $\rho \cdot {\bf z}$ is maximized. Recall that the face $F$ is dual to the cone $\sigma_F \in \mathrm{Fan\ Newt\ }P$ in whose interior lies $\rho$, via \eqref{eq:coneface}.

Together, the value of $\mathrm{Trop\ }P(\rho)$ and the initial form $P|_\rho$ control the leading behavior of $P$ under the rescaling $\alpha \to \alpha \lambda^{-\rho}$, in the limit $\lambda \to 0$, which motivates us to understand them better.
It turns out that they are respectively piecewise linear and piecewise constant on the fan $\Sigma = \rm{Fan\ Newt\ } P$.
To see this, consider a cone $\sigma \in \Sigma$. For any $\rho$ in the interior $\sigma^+$ of the cone, the functional $\rho \cdot {\bf z}$ is always maximized over $\mathrm{Newt\ } P$ at the same face $F_\sigma$ of the Newton polytope. This shows that $P|_\rho$ is constant as $\rho$ varies in $\sigma^+$. Furthermore, over $\sigma$ we can compute $\mathrm{Trop\ }P(\rho)$ as $\rho \cdot {\bf m}$ for \emph{any} ${\bf m}$ on the face $F_\sigma$,
\begin{align}
    \rho \in \sigma \Rightarrow \mathrm{Trop\ } P(\rho) = \rho \cdot {\bf m}, \quad \forall {\bf m} \in F_\sigma,
    \label{eq:troponcone}
\end{align} 
which shows the linearity of $\mathrm{Trop\ } P$ as a function of $\rho \in \sigma$.
Having identified the domains of linearity of the function $\mathrm{Trop\ } P$ with the cones of $\Sigma$, it remains to understand its values there. By piecewise linearity it is sufficient to know the values on the generators of the one-dimensional cones $\Sigma(1)$. Recall that these correspond to the rows of the matrices $\rho_f$ and $\rho_e$ of the H-presentation of $\mathrm{Newt\ } P$. Say that $d_\rho - \rho \cdot {\bf z}$ is one of the affine functions appearing in the presentation, then $\mathrm{Trop\ }P(\rho) = d_\rho$.

Lower dimensional polytopes play a special role in the construction of our subtraction scheme, for the following reason.
Suppose that $d - \rho \cdot {\bf z}$ is an equation of the H-presentation of $\mathrm{Newt\ } P$. Then 
\begin{align}
    P(\alpha \lambda^{-\rho}) = \lambda^{-\mathrm{Trop\ }P(\rho)} P(\alpha).
    \label{eq:factorization}
\end{align}
In other words, after the rescaling \eqref{eq:rescaling} the dependence on $\lambda$ factors perfectly.
Furthermore, 
\begin{align}
    \mathrm{Trop\ }P(\rho + \rho') = \mathrm{Trop\ }P(\rho) + \mathrm{Trop\ }P(\rho'),
    \label{eq:linearity}
\end{align} that is $\mathrm{Trop\ }P$ is linear in the direction of $\rho$.

Let us think of the act of taking initial forms as an operator $P \to P|_\rho$ acting on polynomials.
Suppose that $\rho$ and $\rho'$ are vectors compatible with respect to the fan $\Sigma = \mathrm{Fan\ }\mathrm{Newt\ } P$. Then we have $P|_\rho|_{\rho'} = P|_{\rho'}|_{\rho}$. This is obvious since the initial form $P|_{\rho}$ is given by keeping only the monomials $s_{\bf m} \alpha^{\bf m}$ of $P$ whose exponent vectors ${\bf m}$ are on the face where $\rho$ is maximized. If $\rho$ and $\rho'$ are compatible, they are maximized at a common face, and it does not matter in which order we restrict there. This justifies introducing the notation $P|_{\mathrm{Span}_+ \{\rho,\rho'\}}$.
\begin{example}
    Consider the polynomial $P = 1 + x + y + x y$. Its Newton polytope and normal fan are shown in Fig.\ref{fig:squaregeometry} (a) and (b).
    We have $P|_{(-1,0)}|_{(0,-1)} = (1+y)|_{(0,-1)} = 1 = P|_{(0,-1)}|_{(-1,0)}$.
    On the other hand, $P|_{(-1,0)}|_{(1,0)} = 1+y \ne x(1+y) = P|_{(1,0)}|_{(-1,0)}$.
\end{example}
\begin{figure}
    \centering
    \begin{subfigure}[b]{.3\textwidth}
        \centering        \includegraphics[width=\textwidth]{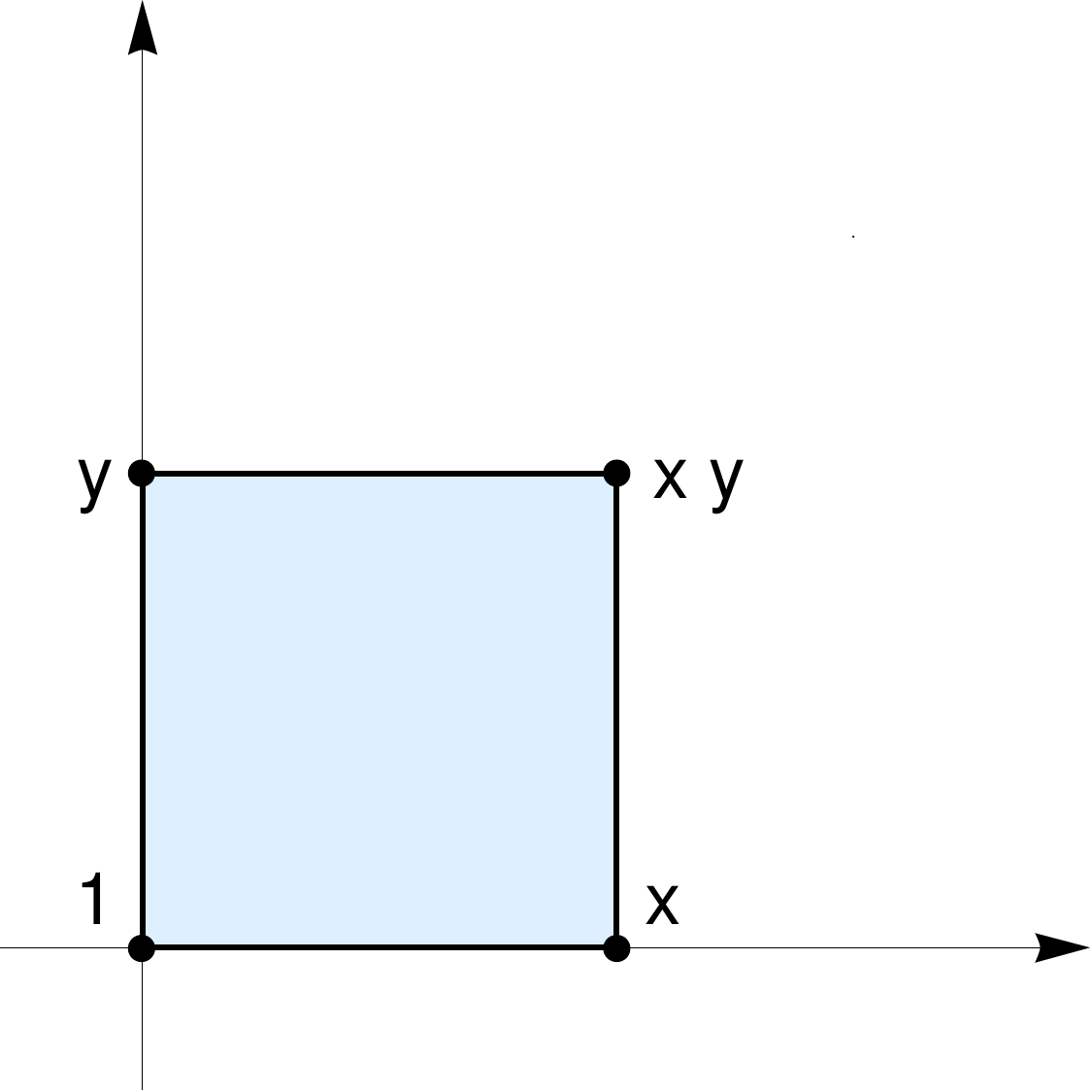}
        \caption{}
    \end{subfigure}%
    \begin{subfigure}{0.3\textwidth}
        \centering
        \includegraphics[width=\textwidth]{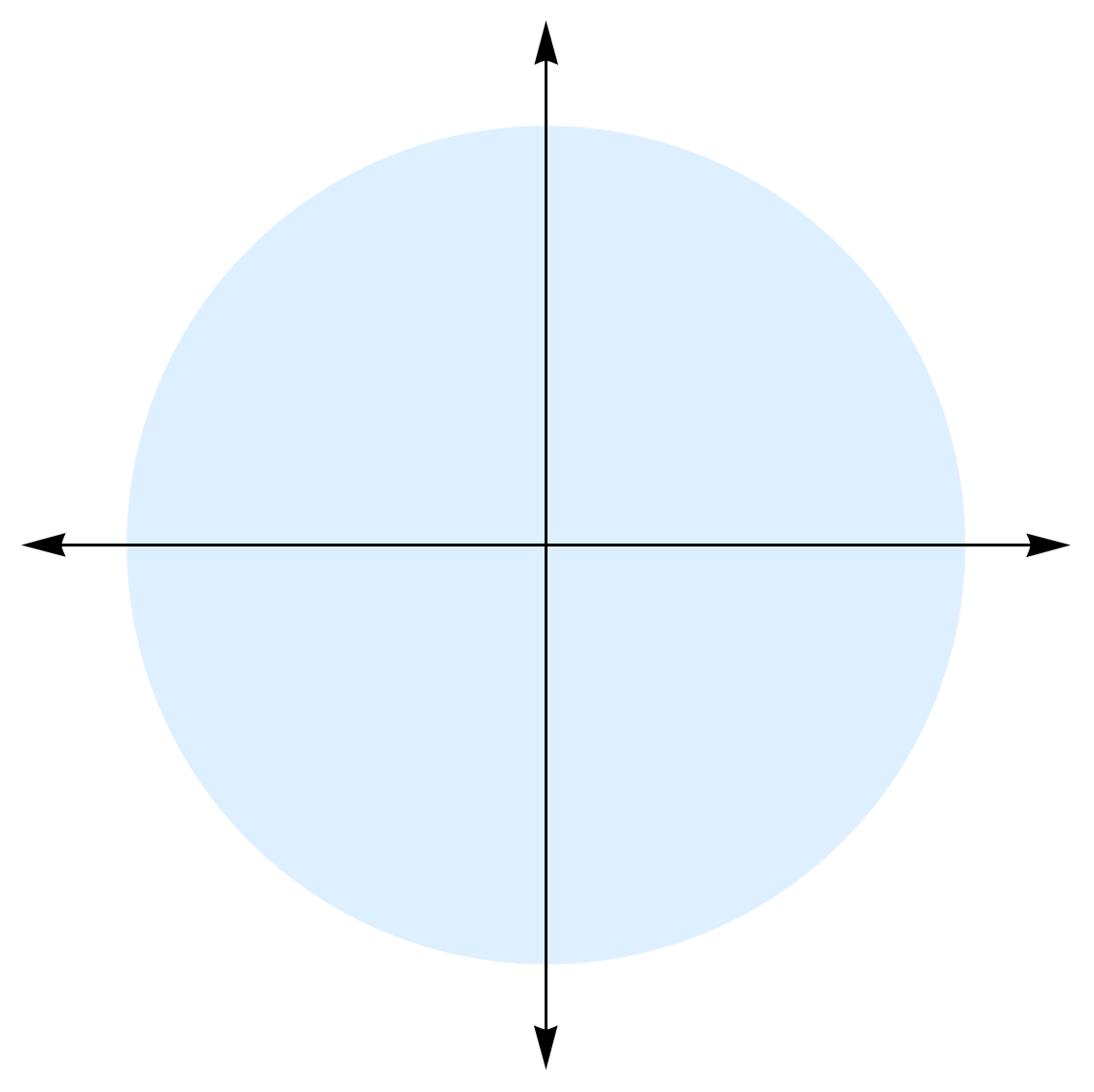}
        \caption{}
    \end{subfigure}
    \caption{The Newton polytope of $1+x+y+x y$ (a) and its normal fan (b). }
    \label{fig:squaregeometry}
\end{figure}

Consider a cone $\sigma$ of the fan $\Sigma = \mathrm{Fan\ Newt\ } P$. The initial form of $P$ with respect to $\sigma$, $P|_\sigma$, is the polynomial obtained from $P$ by keeping only the monomials with exponent vectors on that face $F_\sigma$ of $\mathrm{Newt\ } P$ which is dual to $\sigma$.
In particular, the polytope $\mathrm{Newt\ } P|_\sigma$ is precisely the face $F_\sigma$, and therefore its H-presentation is obtained from that of $\mathrm{Newt\ } P$ by turning finitely many inequalities into equations (and removing any redundant inequalities).
It follows that the tropicalization of $P$ and of its initial forms are tightly related, as expressed by the following lemma.
\begin{lemma}
\label{th:restrictionpoly}
Let $\sigma' \in \Sigma|_\sigma$ be a cone compatible with $\sigma$. Then $\mathrm{Trop\ } P|_\sigma$ coincide with 
$\mathrm{Trop\ } P$ over $\sigma'$ and extends it linearly over $\sigma' + \mathrm{Span\ Rays\ }\sigma$.
\end{lemma}
\begin{proof}
    Since $\sigma'$ and $\sigma$ are compatible, the corresponding faces meet at $F_{\sigma} \cap F_{\sigma'}$, which is a face both of $\mathrm{Newt\ }P$ and $\mathrm{Newt\ }P|_\sigma$. 
    For any $\rho \in \sigma$, we can evaluate both tropical functions as in \eqref{eq:troponcone}, choosing any ${\bf m}$ in $F_{\sigma} \cap F_{\sigma'}$, which shows that they are equal.
    The rays of $\sigma$ appear as rows of the matrix $\rho_e$ of the H-presentation of $\mathrm{Newt\ } P|_\sigma$. Therefore $\mathrm{Trop\ }P|_\sigma$ extends linearly over $\sigma' + \mathrm{Span\ Rays\ }\sigma$ by \eqref{eq:linearity}.
\end{proof}

\section{Euler Integrals}
\label{sec:euler}

An \emph{Euler integral} is one of the form
\begin{align}
    I({\bf{s}}; \nu,c) = \int_{\mathbb{R}^{n}_{\ge 0}} \frac{d\alpha}{\alpha} \alpha^{\bf \nu} \prod_{i=1}^m P_j(\alpha, s)^{c_j}, \quad \frac{d\alpha}{\alpha} \coloneqq \frac{d\alpha_1}{\alpha_1} \dots \frac{d\alpha_n}{\alpha_n},
    \label{eq:euler2}
\end{align}
where $P_j(\alpha, s)$ are polynomials in the integration variables with coefficients collectively denoted by $\bf s$,
\begin{align}
    P_j(\alpha, s) = \sum_{{\bf m} \in \mathbb{Z}^n} s_{\bf m} \alpha^{\bf m}.
\end{align}
The exponents $\nu = (\nu_1, \dots, \nu_n)$ and $c = (c_1, \dots, c_m)$ are generic complex parameters. We will be interested in studying the expansion of $I({\bf s}; \nu, c)$ as a function of a parameter $\epsilon$ on which all exponents depend affinely.
In the case of Feynman integrals, $\epsilon$ is the \emph{dimensional regularization} parameter.

When studying a Feynman integral, it is useful to assume the coefficients $\bf{s}$ to be complex.
This requires further care in defining the integral, since the integrand is multi-valued and a prescription to choose branches is necessary.
Throughout this paper, we will assume instead the coefficients to be non-negative. Beyond this case, most of the proofs we offer here are not valid anymore. This is a severe restriction on which we will comment at the end. In the case of Feynman integrals, it is expected that the analysis of the divergences in $\epsilon$, as well as the structure of the subtraction scheme will tipically not be affected by this. 

In \eqref{eq:euler2} we have factored out the canonical form $\frac{d\alpha}{\alpha}$. The reason is that it transforms nicely under certain change of variables to be introduced shortly. By Euler \emph{integrand} we mean what multiplies this measure, i.e. 
\begin{align}
    \mathcal{I}({\bf s},\epsilon) = \alpha^{\bf \nu(\epsilon)} \prod_{i=1}^m P_j(\alpha, s)^{c_j(\epsilon)}.
\end{align}
Let us denote by ${\bf P} = \prod_{j=1}^m P_j$ the product of all polynomials appearing in the Euler integrand $\mathcal{I}$ and define the Newton polytope of the integrand $\mathcal{I}$ to be
\begin{align}
    \mathrm{Newt\ } \mathcal{I} \coloneqq \mathrm{Newt\ } {\bf P}.
\end{align}
The Newton polytope is the central geometrical object in the study of an Euler integral\footnote{For the sake of precision, the most important object is really the collection of the polytopes $\mathcal{P}_j = \mathrm{Newt\ } P_j$, however the difference plays a little role in this paper so we ignore it for simplicity.}, because it controls the behavior of the Euler integrand under rescaling
\begin{align}
    \mathcal{I}(\alpha \lambda^{-\rho}) = \lambda^{-\rm{Trop\ }\mathcal{I}(\rho)} \mathcal{I}(\alpha)|_\rho  (1 + \mathcal{J}(\alpha,\lambda)) \quad \mathrm{with } \quad \mathcal{J(\alpha,\lambda)} = \mathcal{O}(\lambda^a),\ a > 0.
    \label{eq:scalingbehavior}
\end{align}
All the ingredients of \eqref{eq:scalingbehavior} are obvious generalizations of their polynomial counterparts discussed earlier,
\begin{align}
    \mathrm{Trop\ } \mathcal{I} &= \nu \cdot t + \sum_{j=1}^m c_j \mathrm{Trop\ } P_j, \label{eq:tropEuler} \\
    \mathcal{I}|_\rho(\alpha) &= \prod_{j=1}^m (P_j|_\rho)^{c_j}, \\
    (1+\mathcal{J}(\alpha,\lambda)) &= \prod_{j=1}^m(1+\mathcal{J}_j(\alpha,\lambda))^{c_j}.
\end{align}
The fan $\Sigma$ is a common refinement of all fans $\mathrm{Fan\ } \mathrm{Newt\ } P_j$. It follows that $\mathrm{Trop\ } \mathcal{I}$ is piecewise linear over $\Sigma$ and that $\mathcal{I}|_\rho|_{\rho'} = \mathcal{I}|_{\rho'}|_{\rho}$ if the rays $\rho$ and $\rho'$ are compatible with respect to $\Sigma$. This justifies the notation $\mathcal{I}|_{\mathrm{Span}_+ \{\rho,\rho'\}}$. 
As an immediate consequence of the above definitions, and of lemma \ref{th:restrictionpoly} we have
\begin{lemma}
\label{th:restrictionint}
Let $\sigma' \in \Sigma|_\sigma$ be a cone compatible with $\sigma$. Then $\mathrm{Trop\ } \mathcal{I}|_\sigma$ coincide with 
$\mathrm{Trop\ } \mathcal{I}$ over $\sigma'$ and extends it linearly over $\sigma' + \mathrm{Span\ Rays\ }\sigma$.
\label{th:restrictionint}.
\end{lemma}
In the study of Euler integrals a special role is played by \emph{monomial changes of variables}, of the form
\begin{align}
    \alpha_i = \prod_{j=1}^n t_j^{- M_{i,j}},
    \label{eq:monomialChange}
\end{align}
or $\alpha = t^{-M}$ in compact notation. In order for the above to define a proper change of variables, we assume that $\det M$ is nonzero.
Therefore, the matrix $-M$ defines a bijective linear map $\phi : \mathbb{R}^n \to \mathbb{R}^n$. The exponent vectors of a monomial changes according to
\begin{align}
    \alpha^{\bf m} \coloneqq \prod_{i=1}^n \alpha_i^{m_i} \to \prod_{i=1}^n \prod_{j=1}^n t_j^{-m_i M_{i,j}} = \prod_{j=1}^n t_j^{-{\bf m} \cdot M},
\end{align}
so that the polytope $\mathcal{P} = \mathrm{Newt\ } \mathcal{I}(\alpha)$ transforms to $\phi(\mathcal{P}) = \mathrm{Newt\ } \mathcal{I}(\alpha(t))$.
If we write $M$ in terms of its columns as $M=(\rho_1| \dots|\rho_n)$, we can think of the change $\alpha = t^{-M}$ as a simultaneous rescaling $\alpha \to \alpha \prod_{j=1}^n t_j^{-\rho_j}$, followed by fixing $\alpha = 1$.
Therefore we can understand the behavior of $\mathcal{I}$ under the change in the same way as we do under rescalings, with all the necessary information encoded in the H-presentation and face lattice of $\mathrm{Newt\ }\mathcal{I}$.

Let us illustrate some essential properties of $I({\bf s}; \epsilon)$ that are immediately read off from the geometrical data of the Newton polytope $\mathcal{P} = \mathrm{Newt\ } \mathcal{I}$.
Suppose that the H-presentation of $\mathcal{P}$ contains the equation $\{d_\rho - \rho \cdot t = 0\}$.
Then the integral that defines the function $I({\bf s}, \epsilon)$ is divergent for any choice of $\epsilon$.
To see this, consider a change of variable $\alpha \to t^{-M}$, using any matrix $M$ having $\rho$ as the first column.
The integral transforms to
\begin{align}
    I({\bf s};\epsilon) = \int_0^\infty \frac{dt_1}{t_1} t_1^{-\mathrm{Trop\ }\mathcal{I}(\rho)}\times \int_{\mathbb{R}^{n-1}} \frac{d (t_2, \dots, t_n)}{(t_2,\dots,t_n)}\mathcal{I}(t_2, \dots, t_n),
\end{align}
note how the integral factorizes due to ${d_\rho - \rho \cdot \alpha = 0}$ being an equation of the H-presentation of $\mathrm{Newt\ } \mathcal{I}$, cfr. with \eqref{eq:factorization}.
The integral factor 
\begin{equation}
    \int_0^\infty \frac{dt_1}{t_1} t_1^{-\mathrm{Trop\ }\mathcal{I}(\rho)},
\end{equation}
is always divergent, regardless of the value of $\mathrm{Trop\ }\mathcal{I}$ on $\rho$.
However, we can give a meaning to it by breaking up the domain of integration as $\mathbb{R} = [0,1] \cup [1,\infty]$. In each domain the integral is convergent for one of the two choices of signs for $z=\mathrm{Trop\ }\mathcal{I}(\rho)$, evaluating to $\pm z$, so that the two contributions sum up to zero\footnote{This is gives rise to interesting phenomena in physical computations, with poles in $\epsilon$ due to infrared and ultraviolet divergences cancelling each other, $\frac{1}{\epsilon_{\rm UV}}-\frac{1}{\epsilon_{\rm IR}} = 0$.}!
In summary, integrands with lower dimensional Newton polytope - which are called \emph{scaleless} in the physics literature - give divergent integral representation of $0$.
Let us now consider an integrand whose Newton polytope is $n$-dimensional. The integral $I({\bf s},\epsilon)$ may of course still be divergent. The positivity assumptions for the coefficients ${\bf s}$ implies that divergence may only originate from regions around the boundaries of the domain of integration, such regions can be exposed by monomial changes of variables. 
If $\mathrm{Trop\ }\mathcal{I}(\rho) \ge 0$ on some ray then the integral is divergent, which can be see again by performing a monomial change of variables associated to a matrix $M$ having $\rho$ among its columns.
The converse is true, if $\mathrm{Trop\ } \mathcal{I} < 0$ for some choice of $\epsilon$, then the integral is convergent \cite{Arkani-Hamed:2022cqe,arkanihamed2020binary}. The proof of this statement is slightly harder, and requires one of the two algorithms which will be reviewed in the rest of this section.

\begin{remark*}
    At this point it should be clear that understanding the facet presentation of the Newton polytope of the Symanzik polynomials $\mathcal{U}$ and $\mathcal{F}$, that appear in the parametric form of Feynman integrals, is a mathematical problem of great physical interest.
    Currently, this presentation is not available if not in special cases \cite{Schultka:2018nrs}. In particular, this means that there is no known \emph{purely graph theoretical} algorithm - i.e. one that avoids solving a linear programming problem - to establish the convergence of an arbitrary Feynman integral.
    \label{rmk:UFpolytopes}
\end{remark*}

The geometrical data of the Newton polytope of $\mathcal{I}$ contains much more sophisticated information about the Euler integral, that allows to formulate two algorithms that we review next: \emph{Nilsson-Passare} analytical continuation and \emph{Sector Decomposition}.

\subsection{Nilsson-Passare analytical continuation}

The \emph{Nilsson-Passare} analytical continuation is an alternative approach to compute the Laurent expansion of an Euler integral with respect to the exponent variables.

Let $\mathcal{I}(\alpha;\epsilon)$ be an Euler integrand with a regulated divergence along some direction $\rho$, that is $\mathrm{Trop\ }\mathcal{I} (\rho) = a + \mathcal{O}(\epsilon)$ with $a\ge 0$.
Consider the shifted Euler integral
\begin{align}
    I(\lambda;\epsilon) = \int_{\mathbb{R}^n_{\ge0}} \frac{d\alpha}{d\alpha} \mathcal{I}(\alpha \lambda^{-\rho};\epsilon),
\end{align}
which reduces to the original integral $I(\epsilon)$ for $\lambda=1$.
In fact, $\frac{d}{d\lambda}I(\lambda;\epsilon)=0$, because the dependence on $\lambda$ can be eliminated at the integrand level by the rescaling $\alpha \to \alpha\lambda^{\rho}$ which leaves invariate both the measure and the integration domain\footnote{In this context, $\lambda$ should be thought of as a parameter, not an integration variable.}.
We can therefore write
\begin{align}
    0 = \left.\left(\frac{d}{d\lambda}I(\lambda;\epsilon)\right)\right|_{\lambda=1} = \int_{\mathbb{R}^n} \frac{d\alpha}{\alpha}\left.\left(\frac{d}{d\lambda} \mathcal{I}(\alpha\lambda^{-\rho};\epsilon)\right)\right|_{\lambda=1} = -\mathrm{Trop\ \mathcal{I}}(\rho) I(\epsilon) + \sum_{i} I_{i}(\epsilon),
\end{align}
or equivalently,
\begin{align}
   I(\epsilon) =  \frac{1}{\mathrm{Trop\ \mathcal{I}}(\rho)} \left(\sum_{i} I_{i}(\epsilon) \right),
\end{align}
where $I_i(\epsilon)$ are again Euler integrals.
It turns out that the integrands $\mathcal{I}_i$ for the integrals $I_i$ satisfy $\mathrm{Trop\ }\mathcal{I}_i \le \mathrm{Trop\ } \mathcal{I}$, with the inequality being strict when both sides are evaluated on the ray $\rho$.
This can be seen by considering how the facet $\rho$ of the Newton polytope of $\mathcal{I}(\alpha \lambda^{-\rho};\epsilon)$ is ``shaved'' by the action of $\frac{d}{d\lambda}$.
By iterating this procedure sufficiently many times, we can express $I(\epsilon)$ as a linear combination of \emph{convergent} Euler integrands multiplied by explicit poles in $\epsilon$\footnote{We note \emph{en-passant} that this procedure proves that Euler integrals have a meromorphic dependence on the exponent variables}.

As anticipated in the introduction, the integrals $I_i$ have the same dimension as $I(\epsilon)$, therefore the Nilsson-Passare approach fails to capture the simplification intrinsic to the computation of the poles in $\epsilon$.
\begin{example}
    Consider the following Euler integrand,
    \begin{align}
        \mathcal{I}(\epsilon) = \left(\frac{\alpha}{1+\alpha}\right)^{\epsilon X_{13}} \left(\frac{1}{1+\alpha}\right)^{1+\epsilon X_{24}},
    \end{align}
    it is logarithmically divergent along the rays $\rho = (- 1)$. The Nilsson-Passare analytical continuation along $\rho$ gives \begin{align}
       I(\alpha) &= \frac{1+(X_{13}+X_{24})\epsilon}{\epsilon X_{13}}\int_{0}^\infty \frac{d\alpha}{\alpha} \alpha^{1+X_{13}\epsilon}(1+\alpha)^{-2-(X_{13}+X_{24})\epsilon} \\
       &= \frac{1}{\epsilon X_{13}} - \frac{(X_{13}+X_{24})^2\epsilon}{X_{13}} + \mathcal{O}(\epsilon^2).
    \end{align}
    \label{ex:np}
\end{example}

Note that the \emph{Nilsson-Passare} analytical continuation can be also used to express a power divergent integral, i.e. one with $\mathrm{Trop\ } \mathcal{I} = a + \mathcal{O}(\epsilon)$ with $a>0$ on some ray $\rho$, in terms of logarithmically divergent integrals, i.e. with $\mathrm{Trop\ } \mathcal{I} = \mathcal{O}(\epsilon)$.

\subsection{Sector Decomposition}

We now describe a well-known algorithm to decompose a Euler integral into integrals with a simpler singularity structure.

Let $\mathcal{I}$ be an Euler integrand and $\Sigma = \mathrm{Fan\ } \mathrm{Newt\ } \mathcal{I}$ the associated fan. If $\Sigma$ is not simplicial, replace it by a simplicial refinement thereof.
The cones of $\Sigma$ provide a triangulation of $\mathbb{R}^n$, which can be identified with the domain of integration of a Euler integrand through the coordinate-wise logarithm map $\alpha \to \rm{Log}(\alpha)$.
Therefore, we can write the Euler integral as follows 
\begin{align}
    I(\epsilon) = \sum_{\sigma \in \Sigma} \int_{\sigma} \frac{d\alpha}{\alpha}\ \mathcal{I}(\epsilon)
\end{align}
We can map each cone to a cube by applying a linear transformation to the positive hyperoctant, followed by the coordinate-wise exponential map.
This can be described directly as a monomial map in the original variables $\alpha$.
Let $\sigma = \mathrm{Span}_+ \{\rho_1, \dots, \rho_n\}$ be a top-dimensional cone of $\Sigma$, list its rays as column vectors of a matrix $M_\sigma = (\rho_1, \dots, \rho_n)$ and consider the monomial change of variables $\alpha = t^{-M_\sigma}$.
The vectors $\{\rho_1, \dots, \rho_n\}$ are compatible with respect to each of the fans $\mathrm{Fan\ } \mathrm{Newt\ } P_j$, therefore the initial forms $P_j|_\sigma$ are all equal to some constant $s_{{\bf m}_j,\sigma}$ (which we take to be $1$ for simplicity).
Therefore, we find
\begin{align}
    I(\epsilon) = \sum_{\sigma \in \Sigma} \det(M_\sigma)\int_{[0,1]^n} \frac{dt}{t} t^{-\mathrm{Trop}\mathcal{I}(\sigma)}\prod_{j=1}^m (1+R_j(t))^{c_j}.
    \label{eq:secdec}
\end{align}
We will refer to the types of integrals appearing on the RHS of \eqref{eq:secdec} as \emph{sector integrals}.
\begin{example}
    Consider again the integral of \ref{ex:np}. The fan of the integrand has two cones, resulting in the decomposition
    \begin{align}
        I = \int_0^1 \frac{dt}{t} t^{\epsilon X_{13}}(1+t)^{-2-\epsilon(X_{13}+ X_{24})} + \int_0^1 \frac{dt}{t} t^{1+X_{24}\epsilon}(1+t)^{-1-\epsilon(X_{13}+X_{24})}.
    \end{align}
\end{example}
Sector integrals have a very special form.
Compared to a Euler integral, the first difference is that the domain of integration is a hypercube, so there is no reason to worry about the behavior of the integrand at infinity. 
But most importantly, \emph{the local behavior at the origin is completely manifested} by the monomial $t^{\bf \nu}$, because the polynomials $1+R_j$ do not vanish at the origin.
Note that both properties are preserved under refinement of $\Sigma$, that is if we were to use any refinement of $\Sigma$ to perform the decomposition.

Because of these properties, understanding local finiteness for sector integrals is much simpler than for Euler integrals.
For every variable $t_i$, let us define the operator $\tau_{t_i}$ which acts on a sector integrand by stripping the monomial $t_i^{\nu_i}$, performs a Taylor expansion of the rest, and then multiplies back in $t_i^{\nu_i}$:
\begin{align}
    \tau_{t_i} t^\nu \prod_{j=1}^m (1+R_j)^{c_j} \coloneqq t^\nu \tau_{t_i} \prod_{j=1}^m (1+R_j)^{c_j},
\end{align}
The result of this operation is a Puiseux power series, i.e. one that involves real exponents.
Due to the non-vanishing of the polynomials $(1+R_j)$ at the origin, the operators $\tau_{t_i}$ commute with each other. This follows by standard results on Taylor series of multi-variate smooth functions.
It follows that a sector integrand admits a well-defined Puiseux power series expansion, which can be obtained by applying sequentially the operators $\tau_{t_i}$ in any order,
\begin{align}
    \mathcal{I} = \tau_{t_1} \dots \tau_{t_n} \mathcal{I} = t^\nu \left(\tau_{t_1} \dots \tau_{t_n} \prod_{j=1}^m (1+R_j)^{c_j}\right) = t^\nu \left(\sum_{{\bf m} \in \mathbb{N}^n} t^{\bf m} {\phi}_{\bf m}(\epsilon)\right).
    \label{eq:expansion}
\end{align}
Let us consider the interplay between expansions in $\epsilon$ and in $t_i$ next.
Suppose that exponents $c_j(\epsilon) = a_j + \mathcal{O}(\epsilon)$. Regardless of the value of $a_j$, we have that $(1+R_j)^{c_j(0)}$ is non-vanishing and therefore all mixed derivatives in $\epsilon$ and in any $t_j$ are well defined. It follows that the operators $\tau_\epsilon$ and $\tau_{t_i}$ commute when acting on $(1+R_j)^{c_j}$.
Therefore the expansion in $\epsilon$ of a sector integrand is given by
\begin{align}
    \tau_\epsilon \mathcal{I} = (\tau_\epsilon t^\nu)\times\left(\tau_\epsilon \prod_{j=1}^m (1+R_j)^{c_j}\right) = (\tau_\epsilon t^\nu) \times \left( \sum_{{\bf m} \in \mathbb{N}^n}t^{{\bf m}} \tau_\epsilon \phi_{{\bf m}}(\epsilon)\right).
    \label{eq:epsexp}
\end{align}

The operators $\tau_{t_i}$, which are defined on individual sector integrands, can be extended by linearity on combination of sector integrands, $\tau_{t} \sum_i \mathcal{I}_i \coloneqq \sum_i \tau_t \mathcal{I}_i$, this is consistent because $\tau_t t^{n} \mathcal{I} = t^n \tau_t \mathcal{I}$, for any integer $n \in \mathbb{Z}$.
All the properties just discussed trivially generalize.

We will say that a combination $\mathcal{I} = \sum_i \mathcal{I}$ of Euler integrands is \emph{renormalized} if  $\tau_{t_i} \mathcal{I} = \mathcal{O}(t)$ for every variable $t_i$. This leads us to the first result on local finiteness
\begin{theorem}
\label{th:localfinite}
If a sector integrand $\mathcal{I}$ is renormalized, then it is locally finite, that is
\begin{align}
    \tau_\epsilon \int_{[0,1]^n} \frac{dt}{t}\mathcal{I} = \int_{[0,1]^n} \frac{dt}{t} \tau_\epsilon \mathcal{I}.
\end{align}
\end{theorem}
To prove this, we take derivatives of the sector integral with respect to $\epsilon$ and bring them under integral sign. We take for granted that if the resulting integrals are convergent, than this is allowed.
The heart of the matter is then to prove the following theorem.
\begin{theorem}
\label{th:localfinite2}
Let $\mathcal{I} = \sum_i \mathcal{I}_i$ be a renormalized sector integrand, then the integrals
\begin{align}
    \int_{[0,1]^n} \frac{dt}{t} \left.\frac{\partial^k}{(\partial \epsilon)^k} \mathcal{I}\right|_{\epsilon=0},
\end{align}
are convergent.
\end{theorem}
\begin{proof}
Due to the positivity assumption on the coefficients of the polynomials, the only problem may arise by integrating around a neighborhood of the origin, where the integrand can be replaced by its series expansion.
For each integrand $\mathcal{I}_i$, separate the Taylor series in \eqref{eq:expansion} into the regular part $R_i$, with ${\bf m} > \lfloor{\bf \nu}\rfloor$ componentwise\footnote{$\lfloor{.}\rfloor$ denotes the integer part}, and the remaining singular part $S_i$.
In order for $\mathcal{I} = \sum_i \mathcal{I}_i$ to be renormalized, it has to be that $\sum_i t^{\nu_i} S_i = 0$.
The derivative $\left.\frac{\partial^k}{(\partial \epsilon)^k} \mathcal{I}\right|_{\epsilon=0}$ is the coefficient of $\epsilon^k$ in the series
\begin{align}
    \tau_\epsilon \mathcal{I} = \sum_i \tau_\epsilon \mathcal{I}_i = \sum_i (\tau_\epsilon t^{\nu_i}) \times \left(\tau_\epsilon S_i + \tau_\epsilon \mathcal{R}_i\right) = \sum_i  (\tau_\epsilon t^{\nu_i}) \times \left(\tau_\epsilon \mathcal{R}_i\right),
\end{align}
which involves terms of the form $\log(t)^{v} t^{v'}$, with $v>0$ and $v'>1$. Therefore, against the measure $\frac{dt}{t}$, it yields convergent integrals.
\end{proof}

\begin{remark*}
    The assumption of the positivity of the coefficients is crucial, otherwise the integrands may develop singularities from region away from the origin.
\end{remark*}

The algorithms based on sector decomposition further provide a canonical way to renormalize any sector integrand $\mathcal{I}$.
They involve using the singular parts $S$ of the expansion $\tau_t \mathcal{I}$, defined in the proof above, as counterterms for a subtraction formula. However, due to its local nature as a series expansion, $S$ cannot define a counterterm on the global space on which the original Euler integrand was defined.
In the rest of the paper we will present an alternative strategy.

\section{Tropicalization}
\label{sec:trop}

Consider a function $f : \mathbb{R}_{+}^n \to \mathbb{C}$. We define its \emph{tropicalization}, $\rm{Trop\ } f : \mathbb{R}^n \to \mathbb{C}$ by
\begin{align}
    \rm{Trop\ } f(t) = \lim_{\alpha' \to 0^+} \alpha' \log |f(e^{-t/\alpha'})|.
\end{align}
If we identify $t$ with the logarithm of $\alpha$, we can think of the tropicalization of a function $f(\alpha)$ as capturing the asymptotic behavior of $f$ at the boundaries of the hyperoctant $\mathbb{R}_{+}^n$.

It is easy to see that if $f$ is a polynomial or a Euler integrand, then $\mathrm{Trop\ } f$ coincides with our anticipated formulae \eqref{eq:tropPoly} and \eqref{eq:tropEuler}.
A different situation is when one considers the tropicalization of a \emph{combination} of Euler integrands.
Let us illustrate this with a simple example, consider
\begin{align}
    \mathcal{I} = (1+\alpha_1+\alpha_2)^\epsilon - (1+\alpha_1)^\epsilon,
\end{align}
then $\rm{Trop\ }\mathcal{I}(0,-1) = -1$. This is different from the result we would get by applying the naive tropical rule of replacing ``sum'' with ``max'', as in $\rm{Trop\ } 1 - x = \rm{max}(0,x)$, that would yield $\rm{Trop\ }\mathcal{I}(0,-1) = 0$.
Indeed, the rules of tropical calculus are
\begin{align}
    {\rm Trop\ } f^\epsilon &= \epsilon\ \mathrm{Trop\ } f \\
    {\rm Trop\ } f+g &= \max(\mathrm{Trop\ } f, \mathrm{Trop\ } g) \\
    {\rm Trop\ } f-g &\le \max(\mathrm{Trop\ } f, \mathrm{Trop\ } g)
    \label{eq:tropcalc}
\end{align}
Clearly, the third rule plays a crucial role in the study of subtraction schemes.

In order to understand the tropicalization of a combination of Euler integrands, one needs to take a step forward and study series expansions of the individual integrands.
For any vector $\rho \in \mathbb{R}^n$, we define an operator $\tau_\rho$ that acts on an Euler integrand returning the \emph{Puiseux series}
\begin{align}
    \tau_\rho \mathcal{I}(\alpha) = \lambda^{-\mathrm{Trop\ } \mathcal{I}(\rho)}\tau_\lambda \lambda^{\mathrm{Trop\ } \mathcal{I}(\rho)}\mathcal{I}(\alpha \lambda^{-\rho}),
\end{align}
where $\tau_\lambda$ is the operator returning the standard Taylor series expansion in $\lambda$.
In practice, to compute $\tau_\rho$ one strips from the rescaled integrand $\mathcal{I}(\alpha \lambda^{-\rho})$ the overall power $\lambda^{- \mathrm{Trop\ } \mathcal{I}(\rho)}$, then performs a Taylor expansion on the rest and finally multiplies back in the factor $\lambda^{- \mathrm{Trop\ } \mathcal{I}(\rho)}$ (in complete analogy with the operators $\tau_z$ described in the Introduction).
The structure of the Puiseaux series is illuminated by the notion of initial form of the integrand:
\begin{align}
    \tau_\rho\ \mathcal{I(\alpha)} = \lambda^{-\mathrm{Trop} \mathcal{I}(\rho)}\tau_\lambda\left[\mathcal{I}|_\rho (1 + \mathcal{J}_\rho(\lambda))\right] = \lambda^{-\mathrm{Trop} \mathcal{I}(\rho)}  \mathcal{I}|_\rho \tau_\lambda (1+\mathcal{J}_\rho(\lambda)),
\end{align}
recall that $\mathcal{J}_\rho(\lambda) = \mathcal{O}(\lambda)$, so that the action of $\tau_\lambda$ on $(1+\mathcal{J}_\rho(\lambda))$ produces a well defined power series.

The operator $\tau_\rho$ is well defined on a single Euler integrand and is then extended by linearity on a combination of integrands.
This is consistent, since $\tau_\rho$ respects the identities originating from any rational functions that may appear as an overall prefactor of an Euler integrand if some of the exponents are rational. This follows from the simple fact
\begin{align}
    \tau_\lambda \lambda^m f(\lambda) = \lambda^m \tau_\lambda f(\lambda),
\end{align}
which is true for any smooth function $f$ and integer $m$.
The operator $\tau_\rho$ allows to compute $\rm{Trop\ }\mathcal{I}(\rho)$, as minus the exponent of the leading term in the Puiseux series $\tau_\rho \mathcal{I}$.

Let us now consider again an Euler integrand $\mathcal{I}$ and two vectors $\rho_1$ and $\rho_2$ that are compatible according to the structure of the fan $\Sigma$ of the integrand, i.e. they belong to a common cone.
Then we have 
\begin{align}
    \tau_{\rho_1} \circ \tau_{\rho_2} \mathcal{I} = \tau_{\rho_2} \circ \tau_{\rho_1 }\mathcal{I}.
\end{align}
The above equation should be interpreted in the following sense. Each operator returns a Puiseux series in its own variable, let them be $\lambda_1$ and $\lambda_2$. Applying both returns a double series and no matter in which order they are applied, one obtains the same series.
To prove it, note that the operations of taking derivatives in $\lambda_1$ and rescaling the variables as $\alpha \to \alpha \lambda_2^{\rho_2}$ clearly commute, so we can manipulate directly the simultaneously rescaled integrand,
\begin{align}
    \mathcal{I}(\alpha \lambda_1^{\rho_1} \lambda_2^{\rho_2}) = \lambda_1^{-\mathrm{Trop\ } \mathcal{I}(\rho_1)}\lambda_2^{-\mathrm{Trop\ } \mathcal{I}(\rho_2)} \mathcal{I}|_{\rho_1,\rho_2}(1 + \mathcal{J}_{\rho_1,\rho_2}).
\end{align}
It follows that $\tau_{\rho_1}$ and $\tau_{\rho_2}$ commute due to the commutativity properties of the initial form and of Taylor expansions of multivariate smooth functions.

\begin{remark*}
    The fact that the sector integrand resulting from the change $\alpha \to t^{-M_\sigma}$ has a well defined series expansion around the origin can be re-interpreted as a consequence of the commutativity of the operators $\tau_\rho$ with $\rho \in \mathrm{Rays\ } \sigma$ for the sector $\sigma \in \Sigma$. The same holds true under refinement of $\Sigma$.
\end{remark*}

\subsection{Local Finiteness for Euler Integrals}

We are now in the position to state and prove a necessary and sufficient criterion for a combination of Euler integrands to be locally finite.

We have the following.
\begin{theorem}
\label{th:localfinite3}
Let $\mathcal{I} = \sum_i \mathcal{I}_i$ be a combination of Euler integrands.
Then $\mathcal{I}$ is locally finite if and only if $\rm{Trop\ } \mathcal{I}(\rho) = a + \mathcal{O}(\epsilon)$, with $a < 0$, for all $\rho \in \mathbb{R}^n$.
\end{theorem}
\begin{proof}
    We apply sector decomposition using a common refinement $\Sigma$ of all the fans $\mathrm{Fan\ Newt\ } \mathcal{I}_i$.
    Thinking of the changes of variables $\alpha \to t^{-M_\sigma}$ as a rescaling followed by $\alpha \to 1$, we see that condition $\rm{Trop\ } \mathcal{I}(\rho) = a + \mathcal{O}(\epsilon)$ implies in each cone the local finiteness condition of Th. \ref{th:localfinite} for sector integrands.
    Noting that the expansion in $\epsilon$ commutes with the monomial changes of variables performed to reduce to sectors concludes the proof.
\end{proof}

Virtually by the same argument, one can also prove the following useful sufficient condition.
\begin{theorem}
\label{th:localfinitesuf}
Let $\mathcal{I} = \sum_i \mathcal{I}_i$ be a combination of Euler integrands and let $\Sigma$ be the common refinement of the fans $\Sigma_i = \rm{Fan\ } \rm{Newt\ } \mathcal{I}_i$.
Suppose that for all $\rho \in \Sigma(1)$,  $\rm{Trop\ } \mathcal{I}(\rho) = a + \mathcal{O}(\epsilon)$ with $a < 0$.
Then $\mathcal{I}$ is locally finite.
\end{theorem}
The advantage over the previous criterion is that the latter only requires to check the value of $\mathrm{Trop\ }\mathcal{I}$ over the finitely many rays of the common refinement, which can be done algorithmically by linear programming software such as \verb|polymake| \cite{assarf2015computing,Gawrilow:2000qhs}.
We stress the importance of working with the common refinement of the fans associated with the individual integrands $\mathcal{I}_i$.
One may wonder whether it is not sufficient to guarantee that $\mathrm{Trop\ } \mathcal{I}$ has the appropriate value on all the rays of each individual fan. However, in general taking common refinements produces rays that are not in the original fans.

Note that while the proofs of the above theorems are based on the idea of decomposing the integrals into sectors, of which there is essentially one per \emph{vertex} of $\mathrm{Newt\ }\mathcal{I}$, the result is that one can check a condition on each \emph{facet} of $\mathrm{Newt\ }\mathcal{I}$.
In general, the number of vertices and of facets of a polytope can be vastly different. For instance, in the case of Feynman integrals\footnote{With some caveat regarding special kinematics already mentioned in \ref{rmk:UFpolytopes}} the relevant polytopes are \emph{generalized permutohedra} \cite{Schultka:2018nrs}, which have roughly $E!$ many vertices against $2^E$ many facets, $E$ being the number of edges of the graph.
Therefore, it is practically useful to know about the above theorems, rather than performing sector decomposition and then apply Th. \ref{th:localfinite} to the resulting sector integrals.

We conclude with a simple useful lemma,\begin{lemma}
\label{lemma:monomials}
Let $\mathcal{I} = \sum_i \mathcal{I}_i$ be a locally finite combination of Euler integrands, multiplying it by any monomial $\alpha^{\bf{m} \epsilon}$ yields a locally finite integrand.
\end{lemma}
A somewhat surprising feature of this fact is that the integrands $\mathcal{I}_i$, individually, may not even be convergent for nonvanishing $\epsilon$, i.e. may not be well defined in dimensional-regularization.
The appearance of such integrands is a common feature in the construction of effective field theories, where it is usually treated by introducing additional regulators.

\section{Subtraction scheme}
\label{sec:scheme}

We now describe a subtraction scheme that can be applied to Euler integrands $\mathcal{I}$ that satisfy a certain property to be introduced shortly. This allows in particular to evaluate the Laurent expansion of the associated Euler integral $I(\epsilon)$, by integrating the Laurent expansions of appropriately renormalized Euler integrands.

Consider an integrand $\mathcal{I}$ and suppose that $\mathrm{Trop\ } \mathcal{I}$ is of order $\mathcal{O}(\epsilon)$ along some rays $\{\rho_i \}_{i \in I}$ of $\Sigma = \rm{Fan\ Newt\ } \mathcal{I}$, and is otherwise of order $a + \mathcal{O}(\epsilon)$ with $a<0$. In physics jargon, we say that $\mathcal{I}$ has only logarithmic singularities as opposed to power divergences, with $a > 0$. By applying the Nilsson-Passare analytical continuation we can always reduce to this case, so there is no loss of generality in this assumption\footnote{Recall that the motivation of the subtraction discussed in this section is to expose the poles in $\epsilon$ while localizing the integral. We cannot employ Nilsson-Passare to express a divergent integral in terms of finite ones, while there is no problem in reducing power-divergent integrals to log-divergent ones.}.
Let us define a new fan $\Sigma^{\rm div}$ whose cones are given by positive spans of pairwise compatible (with respect to $\Sigma$) subsets of $\{\rho_i \}_{i \in I}$. 

We now make a crucial assumption.

\begin{definition}
We say that the geometric property is satisfied by an Euler integrand $\mathcal{I}$ if, for every ray $\rho \in \Sigma^{\rm div}(1)$ and all rays $\rho' \in \Sigma^{\rm div}$ that are compatible with $\rho$, we are able to find a vector $w_{\rho}$ with $w_\rho \cdot \rho' = -\delta_{\rho,\rho'}$, where $\delta$ denotes the usual Kronecker symbol
\label{def:geometricproperty}
\end{definition}
An immediate consequence of the geometric property, and the assumption of logarithmic divergences only, is that $\mathrm{\Sigma}^{\rm div}$ must be simplicial.
If this was not the case, it would be possible to write some ray $\rho \in \mathrm{\Sigma}^{\rm div}$ as a linear combination of compatible rays $\rho'  \in \mathrm{\Sigma}^{\rm div}$, and it would therefore be impossible to find a vector $w_\rho$ satisfying the requirements of the geometric property.
In fact, a stronger statement holds
\begin{lemma}
    Let $\mathcal{I}$ be a logarithmically divergent integrand satisfying the geometric property.
    Any collection of divergent facets $\{\rho_i\}_{i=1}^m$ of $\mathrm{Newt\ }\mathcal{I}$, i.e. with $\mathrm{Trop\ }\mathcal{I}(\rho_i) = \mathcal{O}(\epsilon)$, meet transversely.
    In particular, the cone $\sigma = \mathrm{Span}_+ \{\rho_i\}_{i=1}^m$ of $\Sigma^{\rm div}$ is also a cone of $\Sigma$.
    \label{th:normalintersect}
\end{lemma}
\begin{proof}
    The facets meet in a face $F$ with $\mathrm{codim}(F) \ge m$. Suppose that $\mathrm{codim}(F) > m$, then some other facet inequalities $\{d_k - \rho_k \cdot {\bf z} \le 0\}_{k=1}^{\mathrm{codim}(F)-m}$ of the H-presentation of $\mathrm{Newt\ }\mathcal{I}$ must be saturated when the inequalities $\{d_i - \rho_i \cdot {\bf z} \le 0 \}_{i=1}^m$ are.
    This means that is possible to write a linear relation among compatible facet inequalities of $\mathrm{Newt\ } \mathcal{I}$ of the form,
    \begin{align}
        \sum_{i=1}^m \phi_{i} \left(d_i - \rho_i \cdot {bf z}\right) = \sum_{k=1}^{\mathrm{codim}(F)-m} \phi'_{k} \left(d_k - \rho_k \cdot {bf z}\right) \quad \forall {\bf z},
    \end{align}
    with positive coefficients $\phi_i$ and $\phi'_k$.
    Because the above must hold for all ${\bf z}$, the same linear relation must hold separately among vectors $\rho_i,\rho'_k$, the constants  $d_i,d'_k$ and consequently (by linearity of $\mathrm{Trop\ }\mathcal{I}$ on $\mathrm{Span}_+ \{\rho_i\}_{i=1}^m \cup  \{\rho'_k\}_{k}^{\mathrm{codim}(F)-m}$) among the values $\mathrm{Trop\ }\mathcal{I}(\rho_i),\mathrm{Trop\ }\mathcal{I}(\rho'_k)$.
    Since $\mathcal{I}$ is logarithmically divergent, $\mathrm{Trop\ }\mathcal{I}(\rho_i) = \mathcal{O}(\epsilon)$, this linear relation implies that the values $\mathrm{Trop\ }\mathcal{I}(\rho'_k)$ are also of order $\mathcal{O}(\epsilon)$, i.e. the vectors $\rho'_k$ are in $\Sigma^{\mathrm{div}}(1)$.
    The linear relation among the vectors then implies that $\mathcal{I}$ cannot satisfy the geometric property, which concludes the proof by contradiction.    
\end{proof}
A consequence of this lemma is that two cones $\sigma$ and $\sigma'$ in $\Sigma^{\rm div}$ are compatible with respect
to $\Sigma^{\rm div}$ if and only if they are with respect to $\Sigma$, which ease some of the notation.

We anticipate that the geometric property is in general not satisfied by generic Euler integrals, nor by the specific instances that show up in physics applications.
Therefore, this assumption currently represents the fundamental limitation of our subtraction scheme.
However, note that any integral which does not satisfy the geometric property can always be expressed as a linear combination of integrals that do. This can be done algorithmically by performing a Nilsson-Passare analytical continuation along the rays for which the property does not hold.

From a physical point of view, the geometric property should be thought of as a measure of the intricacy of the singularity structure of a Feynman integral. The richer is the pattern of singular configurations for its internal momenta, the more likely it is that the property will not hold. In particular, highly non-planar massless diagrams are the most likely ones to break this assumption.

From now on, let us assume that $\mathcal{I}$ satisfies the geometric property with vectors $w_\rho$. We use these to define the rational functions,
\begin{align}
    v_{\rho} \coloneqq \frac{1}{1 + \alpha^{w_\rho}},
\end{align}
for every $\rho \in \Sigma^{\mathrm{div}}$.
Let $\sigma_I = \rm{Span_+}\{\rho_i \}_{i \in I} \in \Sigma^{\rm div}$, set
\begin{align}
    v_{\sigma} = \prod_{i \in I} v_{\rho_i}^{1-\mathrm{Trop\ }\mathcal{I}(\rho_i)},
\end{align}
and consider the Euler integrand
\begin{align}
    \mathcal{I}_{\sigma_I} = v_{\sigma_I} \left.\mathcal{I}\right|_{\sigma_I},
    \label{eq:counterterm}
\end{align}
which include $\mathcal{I}$ for $\sigma = {\bf 0}$.
The integrands $\mathcal{I}_\sigma$ are the counterterms of our subtraction scheme.
We begin by studying their tropicalization.
\begin{lemma}
    \label{th:tropcnt}
    Let $\rho \in \mathbb{R}^{n}$ and $\mathcal{I}_{\sigma_I}(\rho)$ as in \eqref{eq:counterterm}, then
    \begin{align}
    \mathrm{Trop\ } \mathcal{I}_{\sigma_I}(\rho) = \left\{ \begin{array}{cc}
         \mathrm{Trop\ } \mathcal{I}(\rho)  &\mathrm{if\ } \rho \in \Sigma^{\rm div}|_{\sigma_I} \\
         a + \mathcal{O}(\epsilon),\quad a< 0 &\mathrm{else} 
    \end{array}  \right.
    \label{eq:tropct}
\end{align}
Furtermore, if $\sigma' = \mathrm{Span}_+ \mathcal{R}\cup r$, with $r \subseteq \mathrm{Rays\ } \sigma_I$, is a cone in $\Sigma^{\mathrm{div}}|_{\sigma_I}$ and $\rho \in \sigma_I^+$ we have that
\begin{align}
    \mathcal{I}_{\sigma'}|_{\rho} = \mathcal{I}_{\mathrm{Span}_+\mathcal{R}}|_{\rho}.
    \label{eq:counterinitial}
\end{align}
\end{lemma}
\begin{proof}
     By the tropical rules \ref{eq:tropcalc}, we have
     \begin{align}
         \mathrm{Trop\ }\mathcal{I}_{\sigma_I} = -\sum_{i \in I}(1 - \mathrm{Trop\ }\mathcal{I}(\rho_i)) \max(0, w_i \cdot \rho) + \mathrm{Trop\ }\mathcal{I}|_{\sigma_I}  \le \mathrm{Trop\ }\mathcal{I}|_{\sigma_I}.
         \label{eq:bound1}
     \end{align}
    For any cone $\sigma' \in \Sigma|_{\sigma_I}$, consider the cone $\gamma_{\sigma'} = \sigma' + \mathrm{Span\ } \mathrm{Rays\ } \sigma_I$.
    We can write the whole space as the union of these (possibly overlapping) cones, $\mathbb{R}^n = \bigcup_{\sigma' \in \Sigma|_{\sigma_I}} \gamma_{\sigma'}$.
    Note that $\sigma_I \in \Sigma^{\rm div}$  is also a cone of $\Sigma$, due to lemma \ref{th:normalintersect}, so that lemma \ref{th:restrictionint} applies to $\mathcal{I}|_{\sigma_I}$.
    This implies that $\mathrm{Trop\ }\mathcal{I}|_{\sigma_I}$ coincides with $\mathrm{Trop\ }\mathcal{I}$ over $\sigma'$ and extends it linearly to $\gamma_{\sigma'}$.
    If $\rho \in \sigma'$ can be written as the positive span of a subset of $\mathrm{Rays\ }\sigma'$ which includes at least one ray not in $\Sigma^{\rm div}(1)$, then $\mathrm{Trop\ }\mathcal{I}|_{\sigma_I}(\rho) = \mathrm{Trop\ }\mathcal{I}(\rho) = a + \mathcal{O}(\epsilon)$, with $a < 0$. By linearity,  $\mathrm{Trop\ }\mathcal{I}_{\sigma_I} = a + \mathcal{O}(\epsilon)$ 
    on $\rho + \mathrm{Span\ }\sigma_I$.
    The bound \eqref{eq:bound1} implies that the same is true for $\mathrm{Trop\ }\mathcal{I}_{\sigma_I}$. 
    Let us therefore assume that $\mathrm{Rays\ }\sigma' \subset \Sigma^{\rm div}(1)$, and if necessary replace $\sigma'$ by $\mathrm{Span\ }_+ (\mathrm{Rays\ } \sigma' \setminus  \mathrm{Rays\ } \sigma_I)$, so that we can write any $\rho$ in $\gamma_{\sigma'}$ uniquely as $\rho = \sum_{i \in I} \mu_i \rho_i + \rho'$, with $\rho' \in \sigma'$.
    The geometric property implies that $w_i \cdot \rho = -\mu_i$, and so the equation in \eqref{eq:bound1} implies $\mathrm{Trop\ } \mathcal{I}_{\sigma_I} = a + \mathcal{O}(\epsilon)$ if $\mu_i < 0$.
    The remaining case is then that $\mu_i \ge 0$, in which case $\rho \in \Sigma^{\rm div}|_{\sigma_I}$ and $\mathrm{Trop\ } \mathcal{I}_{\sigma_I} = \mathrm{Trop\ } \mathcal{I}|_{\sigma_I} = \mathrm{Trop\ } \mathcal{I}$.

    The second part of the lemma follows trivially from noting that if $\rho \in \sigma_I^+$, then $\mathcal{I}|_{\sigma'}|_\rho = \mathcal{I}|_{\mathrm{Span}_+\mathcal{R}}|_\rho$ and $v_{\rho_i}|_\rho = 1$ for all $i \in I$.
    
\end{proof}

Let us now consider the following renormalization map,
\begin{align}
    \mathcal{I} \to \mathcal{I}^{\rm ren} = \sum_{\sigma \in \Sigma^{\rm div}} (-1)^{\mathrm{dim}(\sigma)} \mathcal{I}_\sigma
    \label{eq:ren}
\end{align}
The combinatorics behind the above formula is familiar from the work of BPHZ on UV renormalization \cite{Hepp1966,Zimmermann1969}, and similar \emph{forest formulae} have been written down also for IR divergences \cite{CHETYRKIN1982340, CHETYRKIN1984419, Herzog:2017bjx}.
The novelty compared with these previous results is that the renormalization map defined here produces an integrand that is not only finite but also \emph{locally finite}. 
Furthermore, we will see shortly that our counter-terms can be integrated explicitly as many times as the order of the pole that they cancel, therefore manifesting the pole structure of the integral.

The following is the first main result of this section, 
\begin{theorem}
The integrand $\mathcal{I}^{\rm ren}$ defined in Eq. \eqref{eq:ren} is locally finite.
\label{th:localfiniteren}
\end{theorem}
\begin{proof}
    Due to Th. \ref{th:localfinite3} it suffices to prove that $\mathrm{Trop\ } \mathcal{I}^{\mathrm{ren}}(\rho) = a + \mathcal{O}(\epsilon)$, with $a<0$, for all $\rho \in \mathbb{R}^n$.

    Let $\rho \in \mathbb{R}^n$ be an arbitrary vector.
    The fan $\Sigma$ is complete, therefore the vector $\rho$ must be in some cone of $\Sigma$.
    From lemma \ref{th:tropcnt}, it follows that all counterterm integrands  (including $\mathcal{I} = \mathcal{I}_{\bf 0}$) have tropical value $a + \mathcal{O}(\epsilon)$, with $a < 0$, on any ray which is not in $\Sigma^{\rm div}$, which implies the same for $\mathcal{I}^{\rm ren}$. 
    Therefore, consider the case when $\rho$ belongs to the interior $\sigma_0^+$ of some cone $\sigma_0 \in \Sigma^{\rm div}$.
    The tropical calculus rule \eqref{eq:tropcalc} gives
    \begin{align}
        \mathrm{Trop\ } \mathcal{I}(\rho) &\le \max(\mathrm{Trop\ }\sum_{\sigma \in \Sigma^{\rm div}|_{\sigma_0}}  (-1)^{\mathrm{dim}(\sigma)}\mathcal{I}_{\sigma}(\rho), \max_{\sigma \in \Sigma^{\rm div} \setminus \Sigma^{\rm div}|_{\sigma_0}} \mathrm{Trop\ }\mathcal{I}_\sigma(\rho)).
    \end{align}
    For the cones $\sigma \in \Sigma^{\rm div} \setminus \Sigma^{\rm div}|_{\sigma_0}$, i.e. those not compatible with $\sigma_0$, $\mathrm{Trop\ }\mathcal{I}_{\sigma}(\rho) = a + \mathcal{O}(\epsilon)$, with $a < 0$, by \eqref{eq:tropct}. Therefore, we need to show that also
    \begin{align}
        \mathrm{Trop\ }\sum_{\sigma \in \Sigma^{\rm div}|_{\sigma_0}}  (-1)^{\mathrm{dim}(\sigma)}\mathcal{I}_{\sigma}(\rho) = a + \mathcal{O}(\epsilon).
    \end{align}
     
    We can group the cones $\sigma \in \Sigma^{\rm div}|_{\sigma_0}$ having the same set $\mathcal{R} =\mathrm{Rays\ }\sigma \setminus \mathrm{Rays\ } \sigma_0$.
    This allows to rewrite the sum over the cones in $\Sigma^{\rm div}|_{\sigma_0}$ as
    \begin{align}
        \sum_{\sigma \in \Sigma^{\rm div}|_{\sigma_0}} (-1)^{\mathrm{dim}(\sigma)}\mathcal{I}_{\sigma}(\rho)
        &=
        \sum_{\mathcal{R}} \pm \left( \sum_{r \in 2^{\mathrm{Rays\ }\sigma_0}} (-1)^{|r|}\mathcal{I}_{\mathrm{Span}_+\mathcal{R}\cup r }\right),
        \label{eq:pairs}
    \end{align}
    where $2^X$ denotes the set of subsets of $X$.
    For each admissible value of $\mathcal{R}$, \eqref{eq:counterinitial} implies that the terms in the innermost sum of \eqref{eq:pairs} have initial form on $\rho \in \sigma_0^+$ equal to $(-1)^{|r|}\mathcal{I}_{\mathrm{Span}_+\mathcal{R}}|_{\rho}$.
    Recalling the elementary identity $0=\sum_{r \in 2^X} (-1)^{r}$, it follows that $\tau_\rho$ acts on each term returning a series of order $\lambda^{a+\mathcal{O}(\epsilon)}$ with $a < 0$, which completes the proof.
\end{proof}

To complete the construction of the subtraction scheme, we have to show that we can integrate the counterterms $\mathcal{I}_\sigma$ as many times as the order of the pole in $\epsilon$ that they cancel, while simultaneously expressing the leftover integrations in terms of locally finite integrands.
To streamline the proof, let us introduce some helpful notation.
First, let us extend the definition of the renormalization map to the counterterms,
\begin{align}
    \mathcal{I}_\sigma \to \mathcal{I}_\sigma^{\rm ren} = \sum_{\sigma' \in \Sigma^{\rm div}|_\sigma} (-1)^{\mathrm{dim}(\sigma') - \mathrm{dim}(\sigma)} \mathcal{I}_{\sigma'}
    \label{eq:ren2},
\end{align}
which includes \eqref{eq:ren} as the particular case with $\sigma = \{{\bf 0}\}$.
A simple application of the inclusion-exclusion principle gives the following
\begin{align}
    \mathcal{I} &=
        \sum_{\sigma \in \Sigma^{\rm div}} \mathcal{I}^{\rm ren}_\sigma,
        \label{eq:forest}
\end{align}
The sum on the RHS of \eqref{eq:forest} includes the contribution \eqref{eq:ren}, for $\sigma = \{ \bf 0 \}$.
We now show that each of the remaining contributions, $\mathcal{I}^{\rm ren}_\sigma$, can be integrated exactly $\rm{dim}(\sigma)$ times, expressing the result as a pole $\epsilon^{-\rm{dim}(\sigma)}$ times a locally finite integrand.

Pick a cone, $\sigma = \rm{Span}_+ \{ \rho_1, \dots, \rho_m\}$, contributing to \eqref{eq:forest} and let $\mathcal{J}_\sigma$ be $\mathcal{I}^{\rm ren}_\sigma$ with the prefactor $v_\sigma$ stripped, that is
\begin{align}
    \mathcal{J}_\sigma \coloneqq  \frac{1}{\prod_{i=1}^m v_{\rho_i}^{1 - \mathrm{Trop\ }\mathcal{I}(\rho_i )}}\mathcal{I}^{\rm ren}_\sigma. 
\end{align}
Complete the rays $\rho^{\sigma} = \{ \rho_1, \dots, \rho_m \}$ of $\sigma$ to a basis, by choosing vectors $\eta = \{\eta_{m+1}, \dots ,\eta_{n}\}$, and consider the matrix $M = (\rho^{\sigma}|\eta)$ having the basis vectors as columns.
We use $M$ to build the monomial change of variables,
\begin{align}
    \alpha_i \to \prod_{j=1}^n t_i^{-M_{i,j}},
    \label{eq:intchange}
\end{align}
under which $\mathcal{J}_\sigma$ transforms to
\begin{align}
    \mathcal{J}_\sigma(t) = \prod_{i=1}^m t_i^{-\mathrm{Trop\ } \mathcal{I}(\rho_i)} \times \tilde{\mathcal{J}}_\sigma(t_{m+1},\dots,t_n),
\end{align}
note that the dependence on the variables $(t_1, \dots, t_m)$ factors out. This is ultimately due to $\mathcal{I}^{\rm ren}_\sigma$ being built from initial forms of $\mathcal{I}$ and to \eqref{eq:factorization}.
The dependence on the remaining variables is through the integrand $\tilde{\mathcal{J}}_\sigma(t_{m+1},\dots,t_n)$, which is of the same form as \eqref{eq:ren} with the vectors $w_\rho$ replaced by $\tilde{w}_\rho = w_\rho \cdot (-M)$.
Recall from our discussion in Section \ref{sec:euler} that under the change \eqref{eq:intchange} the $\mathrm{H}$-presentation of the Newton polytope of the initial form $\mathcal{I}|_{\sigma}$ transforms as
\begin{align}
    \mathrm{Poly}(d_f | \rho_f, d_e, \rho_e) \to \mathrm{Poly}(d_f | \tilde{\rho}_f = \rho_f \cdot (-M)^{-1}, d_e | \tilde{\rho}_e = \rho_e \cdot (-M)^{-1}).
\end{align}
so that the condition $\tilde{w}_\rho \cdot \tilde{\rho}'  = -\delta_{\rho,\rho'}$ is preserved. It follows that $\tilde{\mathcal{J}}_\sigma$ satisfies the geometric property with vectors $\tilde{w}_\rho$ and it is locally finite by Theorem \ref{th:localfiniteren}.
On the other hand, the integration in the $(t_1,\dots,t_m)$ variables can be performed explicitly,
\begin{align*}
    &\det(M) \int \frac{dt_1}{t_1}\dots\frac{dt_m}{t_m} \prod_{i=1}^m \frac{(t_i)^{-\mathrm{Trop\ } \mathcal{I}(\rho_i)}}{(1 + t^{w_{\rho_i} \cdot M})^{1 - \mathrm{Trop\ } \mathcal{I}(\rho_i) }} = \left[ w_{\rho_i} . \rho_j = - \delta_{i,j}, u_i \coloneqq w_{\rho_i} . \eta  \right] =
    \\ &\det(M)\int \frac{dt_1}{t_1}\dots\frac{dt_m}{t_m}  \prod_{i=1}^m \frac{(t_i)^{-\mathrm{Trop\ } \mathcal{I}(\rho_i)}}{(1+ t_i (t_{m+1},\dots,t_n)^{u_i})^{1 - \mathrm{Trop\ } \mathcal{I}(\rho_i) }} =\\
    &\frac{\det(M)}{\prod_{i=1}^m -\mathrm{Trop\ } \mathcal{I}(\rho_i)}\prod_{i=1}^m (t_{m+1},\dots,t_n)^{- \sum_{i=1}^m u_i \mathrm{Trop\ } \mathcal{I}(\rho_i)}.
\end{align*}
Therefore we have,
\begin{align}
    \int \frac{d\alpha}{\alpha}\mathcal{I}^{\rm ren}_\sigma = \frac{\det (M)}{\prod_{i=1}^m -\mathrm{Trop\ } \mathcal{I}(\rho_i)} \int \frac{dt_{m+1}}{t_{m+1}} \dots \frac{dt_n}{t_n} \prod_{i=1}^m (t_{m+1},\dots,t_n)^{- \sum_{i=1}^m u_i \mathrm{Trop\ } \mathcal{I}(\rho_i)} \tilde{\mathcal{J}},
\end{align}
the remaining integrands is locally finite due to $\tilde{\mathcal{J}}$ being locally finite, lemma \ref{lemma:monomials} and the fact that $\mathrm{Trop\ } \mathcal{I}(\rho_i) = \mathcal{O}(\epsilon)$. 

In the above, any choice $\eta$ for the completion of $\rho$ to a basis does the job, but in practice it is useful to choose a subset of the standard unit vectors.
To do so, choose any non-vanishing $m$-minor of the matrix $(\rho_1 | \dots | \rho_m)$, say that is formed by the rows $J^{c}=(1,\dots,n)\setminus J$, then choose $\eta = \{e_j\}_{j \in J}$.
With this choice, $\det(M) = (\rho_f)_{I,J}$, where $\rho_f$ is the matrix appearing in the facet presentation of $\mathrm{Newt\ }\mathcal{I}$. But more importantly,  $\tilde{\mathcal{J}}_\sigma = \mathcal{I}^{\rm ren}|_\sigma$ with $\alpha_i = 1$ for $i \in J^{\rm c}$ and $\alpha_j$ playing the role of $t_j$ if $j \in J$.
In particular, note that if $\mathcal{I}$ is linear in every integration variable, the same is true for $\tilde{\mathcal{J}}_\sigma$, this is helpful in connection with the use of \verb|HyperInt| to integrate the locally finite integrands.

It is convenient to summarize the results of this section in a master theorem,
\begin{theorem}[Subtraction Formula]
    \label{th:subformula}
    Consider a logarithmically divergent Euler integrand $\mathcal{I}$ satisfying the geometric property \ref{def:geometricproperty} with vectors $w_\rho$.
    Let $\Sigma = \mathrm{Fan\ Newt\ }\mathcal{I}$ be its normal fan, and $\Sigma^{\rm div}$ be the fan formed by taking positive spans of those compatible rays of $\Sigma$ on which $\mathrm{Trop\ } \mathcal{I} = \mathcal{O}(\epsilon)$. 
    For each cone $\sigma \in \Sigma^{\rm div}$, including ${\bf 0}$, define     
    \begin{align*}
        v_{\sigma,w} = \prod_{\rho \in \mathrm{Rays\ }\sigma}\left(\frac{1}{1+\alpha^{w_\rho}}\right)^{1 - \mathrm{Trop\ } \mathcal{I}(\rho)},
    \end{align*} 
    and
    \begin{align*}
        \mathcal{I}^{\rm ren}_\sigma = v_\sigma^{-1} \sum_{\sigma' \in \Sigma^{\rm div}|_\sigma} (-1)^{\mathrm{dim}(\sigma')} v_{\sigma'} \mathcal{I}|_{\sigma'}.
    \end{align*}
    List the vectors $\mathrm{Rays\ }\sigma$ as columns of a matrix $\rho^\sigma$, and choose an index $J_\sigma \subseteq (1,\dots,n)$ of cardinality $n-m$, such that its complement, $J_\sigma^{c} = (1,\dots,n)\setminus J_\sigma$, labels a non-zero minor $\det(\rho^\sigma_{J_{\sigma}^c})$ of $\rho^\sigma$.
    Let $\pi_\sigma: \mathbb{R}^n \to \mathbb{R}^{n-m}$ be the linear map that keeps only the entries labelled by $J_\sigma$. 
    Let $u_\sigma = \sum_{\rho \in \mathrm{Rays\ } \sigma} \mathrm{Trop\ } \mathcal{I}(\rho)\  \pi_\sigma(w_\rho)$. 
    Set $\mathrm{Vol}(\sigma) \coloneqq \det(\rho^{\sigma}_{J^{c}_\sigma}) \prod_{\rho \in \mathrm{Rays\ }\sigma}(-\mathrm{Trop\ } \mathcal{I}(\rho))^{-1}$.
    Then
    \begin{align}
        \int_{\mathbb {R}^n_{+}} \frac{d\alpha}{\alpha}\mathcal{I} =\sum_{\sigma \in \Sigma^{\rm div}} \mathrm{Vol(\sigma)} 
        \int_{V_\sigma} \frac{d\alpha}{\alpha} \alpha^{u_\sigma}\mathcal{I}_\sigma^{\rm ren},
    \end{align}
    where $V_\sigma = \{\alpha_{j} \ge 0\}_{j \in J_\sigma} \cap \{\alpha_i = 1\}_{i \in J_\sigma^c}$ and all integrands are locally finite.
\end{theorem}
Let us remark that any Euler integrand that does not satisfy the assumptions of Th. \ref{th:subformula}, either because it is power-divergent or because it does not satisfy the geometric property, can always be expressed as a collection of Euler integrands that do satisfy the assumptions. This can be algorithmically done by applying the Nilsson-Passare analytical continuation along the rays that break the assumptions above. Therefore, the subtraction formula can be applied to any Euler integral.

In the next section, we will illustrate explicitly this subtraction formula in a variety of examples.
We conclude by noting that our subtraction formula does more than merely providing the Laurent expansion for an Euler integral: it keeps separated the contributions to a given order in $\epsilon$ that come from different regions in the integration space.
In the case of Feynman integrals, this has the appealing feature of keeping separated poles that have an UV, IR or mixed origin. 

\section{Examples}
\label{sec:examples}

We now present explicit examples of our subtraction scheme, drawing from a range of physics applications.
All explicit results are obtained by integrating locally finite integrands with \verb|HyperInt|, which is valid thanks to built-int the shuffle regularization procedure of this algorithm.

\subsection{One loop triangle}

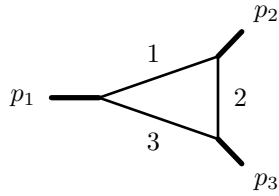
\begin{figure}[h!]
\centering
\begin{fmffile}{first-diagram}
\begin{fmfgraph*}(80,50)
    \fmfkeep{first-diagram}
    \fmfpen{thin}
    \fmfleft{e1}
    \fmfright{e3,e2}
    \fmflabel{$p_1$}{e1}
    \fmflabel{$p_2$}{e2}
    \fmflabel{$p_3$}{e3}
    \fmf{plain,tension=.2, label=$1$}{v1,v2}
    \fmf{plain,tension=.2, label=$2$,label.side=left}{v2,v3}
    \fmf{plain,tension=.2, label=$3$,label.side=left}{v3,v1}
    \fmf{plain, width=thick}{e1,v1}
    \fmf{plain, width=thick}{e2,v2}
    \fmf{plain, width=thick}{e3,v3}
\end{fmfgraph*}
\end{fmffile}
\vspace{.1cm}
\caption{A triangle diagram}
\label{fig:trianglediagram}
\end{figure}
Consider the diagram of Fig. \ref{fig:trianglediagram}.
In the frame where $\alpha_3 = 1$, the corresponding integral is given by
\begin{align*}
    I_G(p_1^2,p_2^2,p_3^2;\epsilon) &=
    \Gamma(d_G) \int_0^\infty \frac{d\alpha}{\alpha} \mathcal{I}_G(p_1^2,p_2^2,p_3^2;\epsilon) \\
    &=
    \Gamma(d_G) \int_0^\infty \frac{d\alpha}{\alpha} \alpha \left(1+\alpha_1+\alpha_2\right)^{d_G-\frac{D}{2}}\left(\alpha_1 p_1^2 + \alpha_1 \alpha_2 p_2^2+ \alpha_2 p_3^2\right)^{-d_G},
\end{align*}
where $d_G = 3 - D/2$. 
Let us choose kinematics $p_2^2=p_3^2=0$ and $D=4-2\epsilon$ space-time dimensions.
In Fig. \ref{fig:trianglegeometry} is shown the polytope $\mathrm{Newt\ } \mathcal{I}(p_1^2, 0,0;\epsilon)$ and its normal fan.
\begin{figure}[h!]
    \centering
    \begin{subfigure}[b]{.3\textwidth}
        \centering        \includegraphics[width=\textwidth]{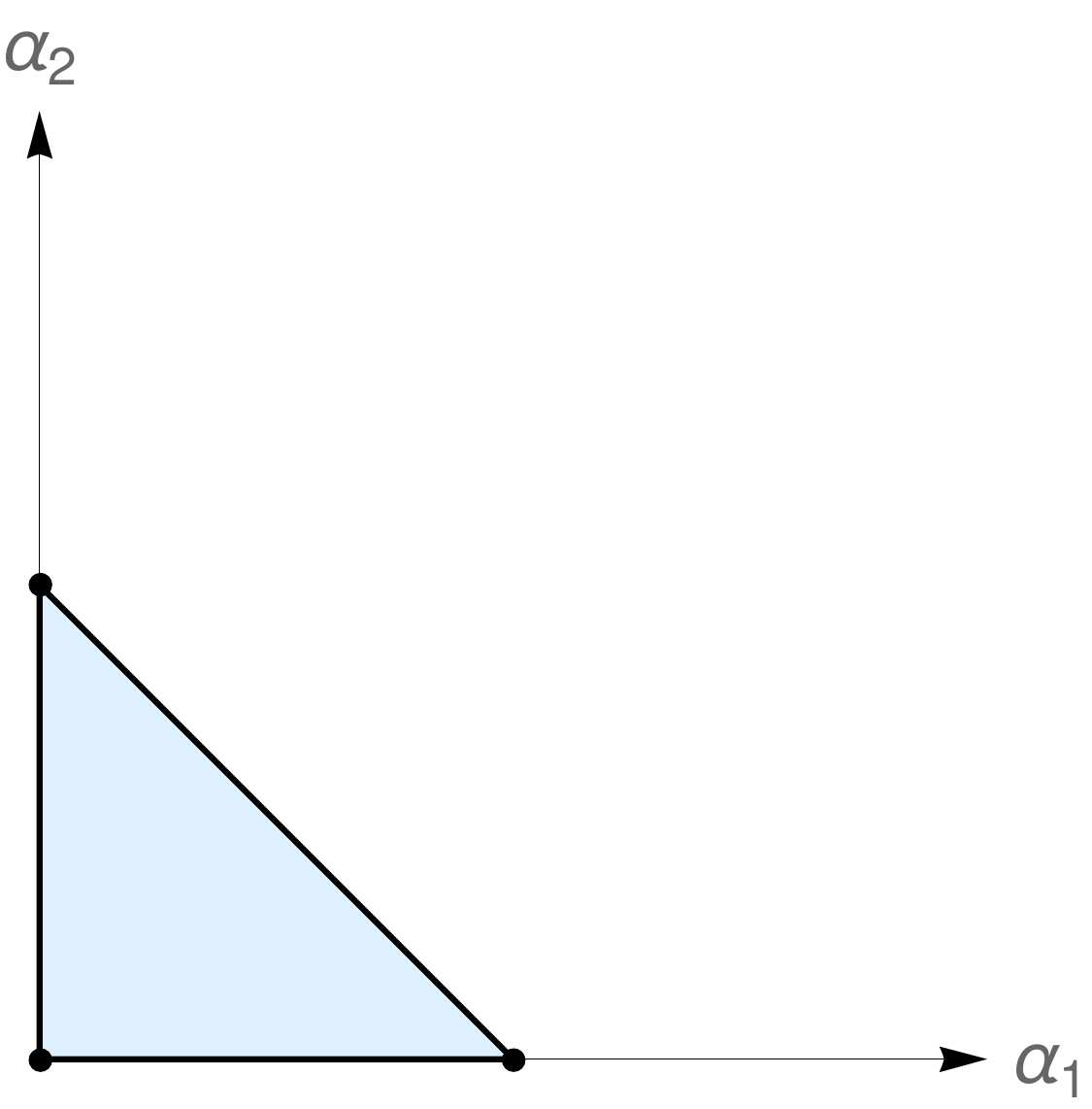}
        \caption{}
    \end{subfigure}%
    \begin{subfigure}{0.3\textwidth}
        \centering
        \includegraphics[width=\textwidth]{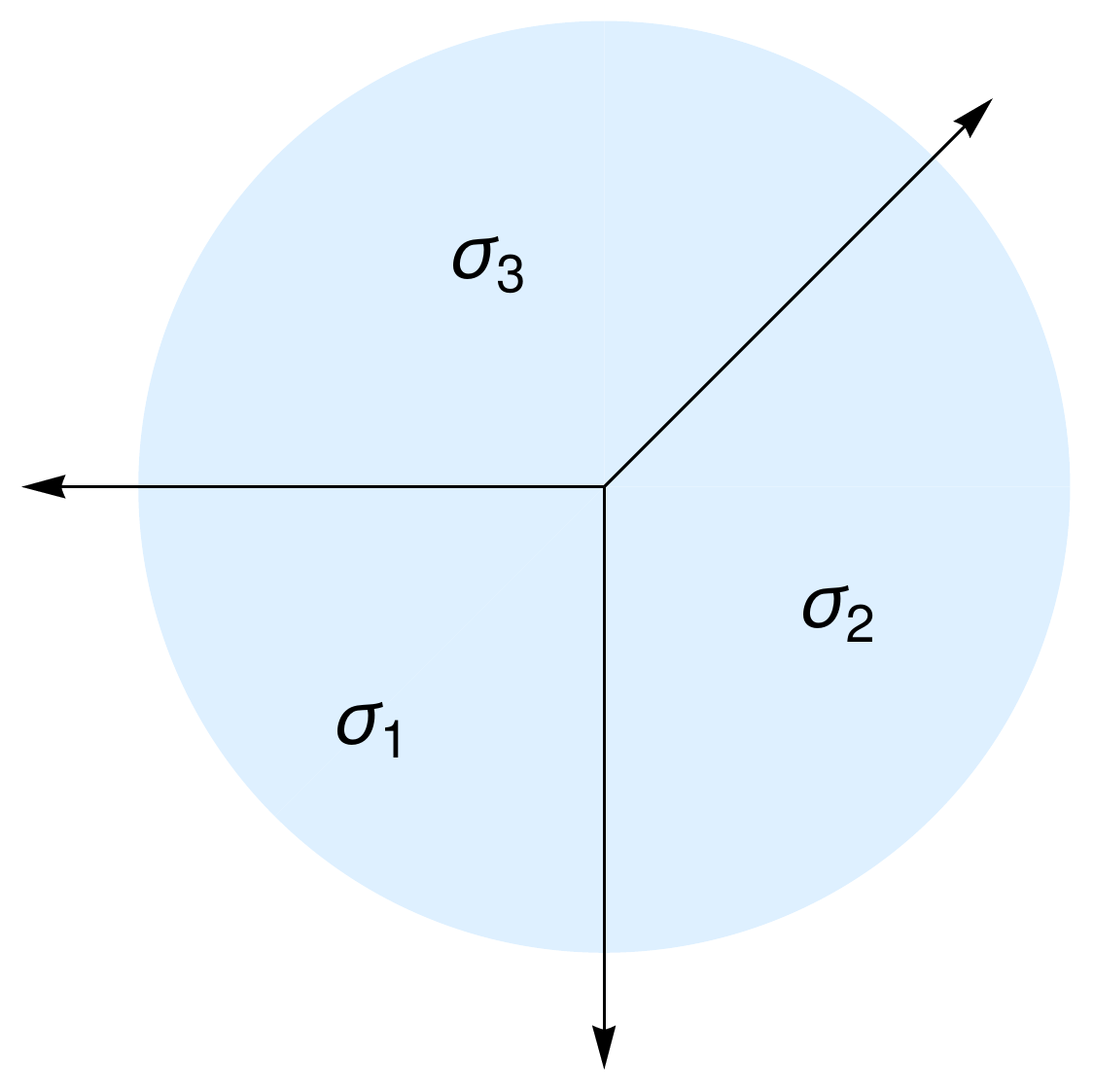}
        \caption{}
    \end{subfigure}
    \begin{subfigure}[b]{0.3\textwidth}
        \centering
        \includegraphics[width=\textwidth]{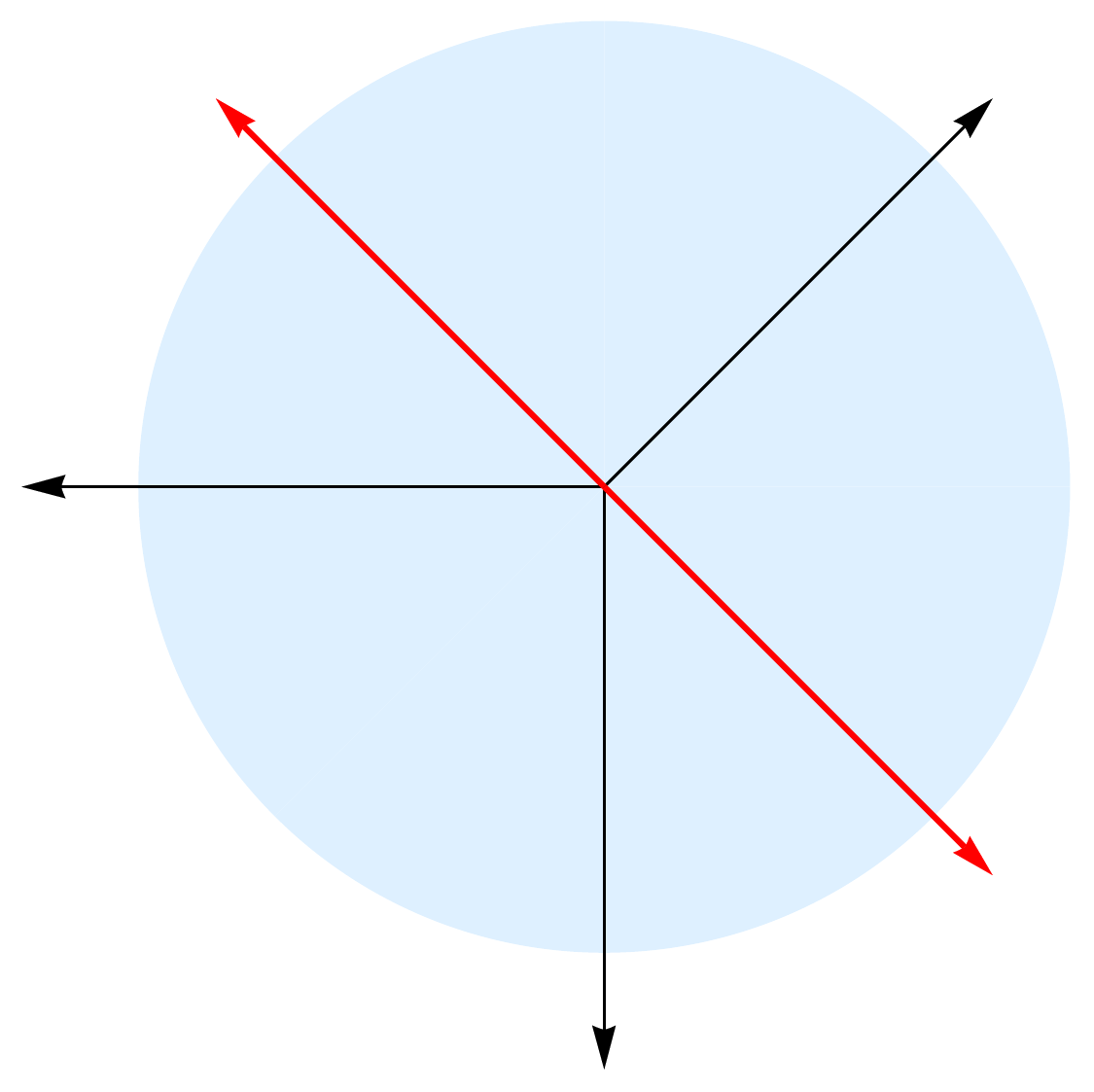}
        \caption{}
    \end{subfigure}
    \caption{The Newton polytope (a) and its normal fan (b). The refinement associated to the ``bad'' subtraction introduces two new rays, drawn in red in (c).}
    \label{fig:trianglegeometry}
\end{figure}
 The tropicalization $\mathrm{Trop\ } \mathcal{I}_G(p_1^2,0,0;\epsilon)$ is of order $\mathcal{O}(\epsilon)$ on the rays $\rho_{1} = (-1,0)$ and $\rho_3 = (1,1)$. The divergent subfan is therefore 
\begin{align}
    \Sigma^{\rm div} = \left\{ \rho_1, \rho_3, \sigma_{3} = \mathrm{Span}_+\{\rho_1,\rho_3\}\right\}.
\end{align}
The integrand $\mathcal{I}_G(p_1^2,0,0;\epsilon)$ satisfies the geometric property with vectors 
\begin{alignat*}{3}
    &w_{\rho_1}&&=(-1,1)\quad &&\Rightarrow\quad  v_{\rho_1} = \frac{1}{1+\alpha_2/\alpha_1} \\
    &w_{\rho_3}&&=(0,1)\quad &&\Rightarrow\quad v_{\rho_1} = \frac{1}{1+\alpha_2}
\end{alignat*}  
We are therefore in position to apply Th. \ref{th:subformula}.
\begin{table}[h!]
\renewcommand{\arraystretch}{1.5}
\centering
\begin{tabular}{ |c c c c c| } 
 \hline
 $\sigma \in \Sigma^{\rm div}$ & $\mathrm{Vol}(\sigma)$ & $\mathcal{I}|_\sigma$ & $V_\sigma$ & $u_\sigma$ \\ 
 \hline
 $\bf 0$ & $1$ & $(p_1^2)^{-(1+\epsilon)}\alpha_1^{-\epsilon} \alpha_2 (1+\alpha_1+\alpha_2)^{-1+2\epsilon}$ & $(\alpha_1,\alpha_2)$ & (0,0) \\ 
 \hline
 $\rho_1$ & $\epsilon$ & $(p_1^2)^{-(1+\epsilon)}\alpha_1^{-\epsilon} \alpha_2 (1+\alpha_2)^{-1+2\epsilon}$ & $(1,\alpha_2)$ & (0,1) \\ 
 \hline
 $\rho_3$ & $\epsilon$ & $(p_1^2)^{-(1+\epsilon)}\alpha_1^{-\epsilon} \alpha_2 (\alpha_1+\alpha_2)^{-1+2\epsilon}$ & $(\alpha_1,1)$ & (1,0) \\ 
 \hline
 $\mathrm{Span}_+\{\rho_1,\rho_3\}$ & $\epsilon$ & $(p_1^2)^{-(1+\epsilon)}\alpha_1^{-\epsilon} \alpha_2^{2\epsilon}$  & $(1,1)$ & (0,0) \\ 
 \hline
\end{tabular}
\caption{The data to build the subtraction scheme}
\label{tab:data}
\end{table}
Table \ref{tab:data} gathers all the data necessary to build the subtraction scheme. We have
\begin{align*}
    \frac{1}{\Gamma(d_G)} I_G(p_1^2,0,0;\epsilon) &= 
    \int_{V_{\bf 0}} \frac{d\alpha}{\alpha}\left(\mathcal{I} - (v_{\rho_1})^{1+\epsilon} \mathcal{I}|_{\rho_1}-(v_{\rho_3})^{1+\epsilon} \mathcal{I}|_{\rho_3}+(v_{\rho_1})^{1+\epsilon} (v_{\rho_3})^{1+\epsilon}\mathcal{I}|_{\mathrm{Span\ }\{\rho_1,\rho_3\}}\right) \\
    &+\frac{1}{\epsilon}\int_{V_{\rho_1}} \frac{d\alpha}{\alpha}\alpha^{u_{\rho_1}}\left(\mathcal{I}|_{\rho_1} - (v_{\rho_3})^{1+\epsilon} \mathcal{I}_{\mathrm{Span}_+ \{\rho_1,\rho_3\}}\right)\\
    &+\frac{1}{\epsilon}\int_{V_{\rho_3}} \frac{d\alpha}{\alpha}\alpha^{u_{\rho_3}}\left(\mathcal{I}|_{\rho_3} - (v_{\rho_1})^{1+\epsilon} \mathcal{I}_{\mathrm{Span}_+ \{\rho_1,\rho_3\}}\right)\\
    &+\frac{1}{\epsilon^2} (p_1^2)^{-(1+\epsilon)}.
\end{align*}
The integrals appearing on this formula can be expanded in series of $\epsilon$ directly at the integrand level,
which gives
\begin{align}
    \tau_\epsilon \frac{1}{\Gamma(d_G)} I_G(p_1^2,0,0;\epsilon) = \frac{1}{\epsilon^2}\frac{1}{p_1^2} - \frac{1}{\epsilon}\frac{\log(p_1^2)}{p_1^2} - \zeta(2)\frac{1}{p_1^2} + \frac{\frac{1}{2}\log(p_1^2)^2}{p_1^2} + \mathcal{O}(\epsilon).
\end{align}

\subsection{One loop triangle, bad subtraction}

Let us revisit the previous example choosing now
\begin{alignat*}{3}
    &w_{\rho_1}&&=-(1,1)\quad &&\Rightarrow\quad  v_{\rho_1} = \frac{1}{1+\alpha_1 \alpha_2} \\
    &w_{\rho_3}&&=1(,1)\quad &&\Rightarrow\quad v_{\rho_1} = \frac{1}{1+1/(\alpha_1 \alpha_2)}
\end{alignat*}  
This will show the importance of the assumptions of Th. \ref{th:localfinitesuf}.
Consider the renormalized integrand,
\begin{align}
\mathcal{I}^{\rm ren} = \mathcal{I} - (v_{\rho_1})^{1+\epsilon} \mathcal{I}|_{\rho_1}-(v_{\rho_3})^{1+\epsilon} \mathcal{I}|_{\rho_3}+(v_{\rho_1})^{1+\epsilon} (v_{\rho_3})^{1+\epsilon}\mathcal{I}|_{\mathrm{Span\ }\{\rho_1,\rho_3\}},
\end{align}
we have that $\mathrm{Trop\ } \mathcal{I}^{\rm ren} = a + \mathcal{O}(\epsilon)$ on the rays of $\mathrm{Newt\ } \mathcal{I}$, with $a<0$, so that naively one could expect $\mathcal{I}^{\rm ren}$ to be locally finite.
However, what matters is the \emph{common refinement} of the fans of the integrands appearing in the combination $\mathcal{I}^{\rm ren}$, which include two new rays shown in Fig. \ref{fig:trianglegeometry}.
For instance, we have that $\mathrm{Trop\ } \mathcal{I}^{\rm ren}(1,-1) = \epsilon$, thus invalidating the assumptions of Th. \ref{th:localfinitesuf}.
Consequently, the expansion $\tau_\epsilon \mathcal{I}^{\rm ren}$ yields divergent integrals, and therefore the renormalized integrand is incorrect.

\subsection{Sunrise}

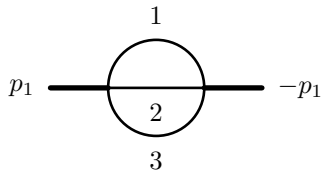
\begin{figure}[h!]
\centering
\begin{fmffile}{diagram-sunrise}
\begin{fmfgraph*}(80,50)
    \fmfpen{thin}
    \fmfleft{e1}
    \fmfright{e2}
    \fmflabel{$p_1$}{e1}
    \fmflabel{$-p_1$}{e2}
    \fmf{plain,tension=.2, label=$1$, left=1}{v1,v2}
    \fmf{plain,tension=.2, label=$2$}{v1,v2}
    \fmf{plain,tension=.2, label=$3$, right=1 }{v1,v2}
    \fmf{plain, width=thick}{e1,v1}
    \fmf{plain, width=thick}{e2,v2}
\end{fmfgraph*}
\end{fmffile}
\vspace{.1cm}
\caption{Sunrise diagram}
\label{fig:sunrisediagram}
\end{figure}
Consider the diagram of Fig. \ref{fig:sunrisediagram}.,
in the gauge with $\alpha_3 = 1$ the corresponding Feynman integral is
\begin{align}
    I(p^2, {\bf m};\epsilon) = \Gamma(3-D) \int_{\mathbb{R}^2_+} \frac{d\alpha}{\alpha} &\alpha \left(\alpha_1 \alpha_2 + \alpha_1 + \alpha_2 \right)^{3(-1+\epsilon)} \nonumber\\
    &\times\left(p^2 \alpha_1 \alpha_2 + (m_1^2 \alpha_1 + m_2^2 \alpha_2 + m_3^2)(\alpha_1 \alpha_2 + \alpha_1 +\alpha_2 \right)^{1-2\epsilon},
    \label{eq:sunriseintegral}
\end{align}
It is well known that in $D=2-2\epsilon$ space-time dimension \eqref{eq:sunriseintegral} is elliptic, and therefore cannot be evaluated in terms of polylogarithms. Geometrically, this is reflected in the fact that the diagram is not linearly reducible.
We will consider the diagram in $D=4-2\epsilon$, in which case the diagram has logarithmic sub-divergences. We list the divergent rays and the corresponding vectors $w_\rho$ as columns of the matrices
\begin{align}
\rho = \left(\begin{matrix}
  1 & 0 & -1 \\
  0 & 1 & -1
\end{matrix}\right) \quad
w = \left(\begin{matrix}
  -1 & 0  & 1 \\ 
   0 & -1 & 0
\end{matrix}\right)
\label{eq:doubleboxdata}
\end{align}
In this example, the Nilsson-Passare analytical continuation would yield finite integrals with poles in $\epsilon$ in the prefactors. Therefore, require expanding the integrals at higher $\epsilon$ orders, which cannot be computed using \verb|HyperInt| due to the underlying lack of linear reducibility in the geometry of the problem - which is dimension-independent.
The subtraction formula instead requires to compute integrals expanded only up to the finite part, which can directly be evaluated via \verb|HyperInt|, obtaining
\begin{align}
    \Gamma(-1+2\epsilon)^{-1} I(p^2, {\bf m};\epsilon) = \frac{\sum_{i=1}^3 m_i^2}{\epsilon} + \frac{s}{2} + \sum_{i=1}^3 m_i^2 - 2 \sum_{i=1}^3 m_i^2 \log(m_i).
\end{align}

This example shows the value of being able to extract only certain parts of a Feynman integral. The full $\epsilon$-expansion of the sunrise in $D=4-2\epsilon$ eventually requires elliptic functions, but if one is interested only in the result up to the finite part, the full complexity of the problem can be avoided, obtaining a result in terms of logarithms only.

\subsection{Double box}

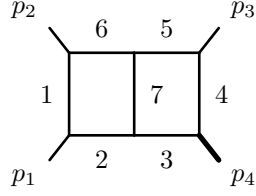
\begin{figure}[h!]
\centering
\begin{fmffile}{diagram-doublebox}
\begin{fmfgraph*}(80,50)
    \fmfpen{thin}
    \fmfleft{e2,e1}
    \fmfright{e3,e4}
    \fmflabel{$p_2$}{e1}
    \fmflabel{$p_1$}{e2}
    \fmflabel{$p_4$}{e3}
    \fmflabel{$p_3$}{e4}
    \fmf{plain,tension=.3, label=$1$}{v1,v2}
    \fmf{plain,tension=.3, label=$2$}{v2,v3}
    \fmf{plain,tension=.3, label=$3$}{v3,v4}
    \fmf{plain,tension=.3, label=$4$}{v4,v5}
    \fmf{plain,tension=.3, label=$5$}{v5,v6}
    \fmf{plain,tension=.3, label=$6$}{v6,v1}
    \fmf{plain,tension=0, label=$7$}{v3,v6}
    \fmf{plain}{e1,v1}
    \fmf{plain}{e2,v2}
    \fmf{plain, width=thick}{e3,v4}
    \fmf{plain}{e4,v5}
\end{fmfgraph*}
\end{fmffile}
\vspace{.1cm}
\caption{Double Box diagram}
\label{fig:doubleboxdiagram}
\end{figure}
Let us now move to a less trivial two-loop example, the one-mass double box diagram of Fig. \ref{fig:doubleboxdiagram}, with 
$p_1^2=p_2^2=p_3^2=0$ and $p_4^2\ne 0$ in $D=4-2\epsilon$.
This type of integral is relevant, for example, in the study of Higgs production through gluon fusion.

In the gauge with $\alpha_7=1$, the Feynman integral reads
\begin{align}
    I(s,t,p_4^2;\epsilon) = \Gamma(3+2\epsilon) \int_{\mathbb{R}^6_+} \frac{d\alpha}{\alpha} \alpha\  \mathcal{U}^{1+3\epsilon} \mathcal{F}^{-(3 + 2 \epsilon)},
\end{align}
with Symanzik polynomials 
\begin{alignat}{2}
&\mathcal{U} = 
 &&\alpha_1 \alpha_3 + \alpha_2 \alpha_3 + \alpha_1 \alpha_4 + \alpha_2 \alpha_4 + \alpha_1 \alpha_5 + \alpha_2 \alpha_5 + \alpha_3 \alpha_6 + \alpha_4 \alpha_6 + \alpha_5 \alpha_6  \notag\\
& && + \alpha_1 + \alpha_2 + \alpha_3 + \alpha_4 + \alpha_5 + \alpha_6  \notag\\
&\mathcal{F} = &&p_4^2 \left(\alpha_1 \alpha_3  \alpha_4 +  \alpha_2 \alpha_3 \alpha_4 + \alpha_3 \alpha_4 \alpha_6 + \alpha_2 \alpha_4 + \alpha_3 \alpha_4\right)  + t \alpha_1 \alpha_4 \notag\\
& &&+ s \left(\alpha_1 \alpha_3 \alpha_5 + 
    \alpha_2 \alpha_3 \alpha_5 + 
    \alpha_2 \alpha_3 \alpha_6 +  
    \alpha_2 \alpha_4 \alpha_6 + \alpha_2 \alpha_5 \alpha_6 + \alpha_3 \alpha_5 \alpha_6 + \right.   \notag\\
& &&\left. +  \alpha_2 \alpha_5 +  \alpha_3 \alpha_5 +  \alpha_2 \alpha_6 +  \alpha_3 \alpha_6 \right) \notag
\end{alignat}
The integrand is divergent along $6$ rays, we list the divergent rays and the associated vectors $w_\rho$ as columns of the following matrices
\begin{align}
\rho = \left(\begin{matrix}
  1 & 1 & 0 & 0 & 0 & -1 \\
  1 & 0 & 0 & 0 & -1 & -1 \\
  0 & 0 & 0 & 0 & -1 & -1 \\
  0 & 0 & -1 & 1 & -1 & 0 \\
  0 & 0 & -1 & 1 & 0 & 0 \\
  0 & 1 & -1 & 0 & 0 & 0
\end{matrix}\right) \quad
w = \left(\begin{matrix}
  0 & -1 & 1 & 0 & 0 & 0 \\
  -1 & 1 & 0 & -1 & 1 & 0 \\
  0 & 0 & 0 & 0 & 1 & 0 \\
  0 & 0 & 0 & 0 & -1 & 0 \\
  0 & 0 & 0 & 1 & -1 & 0 \\
  1 & 0 & 0 & 0 & 0 & -1
\end{matrix}\right)
\label{eq:doubleboxdata}
\end{align},
Let us denote by $(i_1, \dots, i_m)$ the cone $\mathrm{Span}_+ \{\rho_i\}_{i=1,\dots,m}$. Then we have
\begin{align*}
    \Sigma^{\rm div} = \{&(), \\ &(1), (2), (3), (4), (5), (6), \\ &(1, 2), (1, 3), (1, 4), (1, 5), (2, 3), (2, 4), (2, 5), (2, 6), (3, 5), (4, 5), (4, 6), (5, 6),\\ &(1, 2, 3), (1, 2, 4), (1, 2, 5), (1, 3, 5), (1, 4, 5), (2, 3, 5), (2, 4, 5), (2, 4, 6), (2, 5, 6), (4, 5, 6), \\ &(1, 2, 3, 5), (1, 2, 4, 5), (2, 4, 5, 6)\}.
\end{align*}
So that the subtraction formula contains $32$ terms.
This may seem a bit overwhelming, but note that sector decomposition would yield $141$ terms - and break linear reducibility - while Nilsson-Passare would yield a whopping $1497$ number of terms and require expanding integrands up to order $\mathcal{O}(\epsilon^6)$.
\footnote{These numbers have been obtained with a simple-minded private implementation of both algorithms. It is most likely possible to obtain better formulae with more sophisticated tricks}.

We do not write down explicitly the subtraction, although all the necessary data is provided (somewhat implicitly) in \eqref{eq:doubleboxdata}.
For reasons of space, we report the result at a particular kinematical point, $s=t=p_4^2=1$, \footnote{Once again, the result is obtained exactly with {\bf HyperInt}. There is no obstacle in obtaining an analytic result for symbolic values of $p_4^2$, $s$ and $t$.},
\begin{align}
    \Gamma(3+2\epsilon)^{-1} I(1,1,1;\epsilon) = &\frac{1}{\epsilon^4} \frac{1}{2} - \frac{1}{\epsilon^3} \frac{3}{2}   + \frac{1}{\epsilon^2}\frac{14 - \zeta(2)}{4}+ \frac{1}{\epsilon}\frac{-60 + 42 \zeta(2) + 6 \zeta(3)}{8}\\
    &\frac{11160-8820\zeta(2)+1332\zeta(2)^2-1620 \zeta(3)}{720} + \mathcal{O}(\epsilon),
\end{align}
which agrees with the numerical values obtained by \verb|FIESTA| \cite{Smirnov:2021rhf}.

\subsection{Eye Graph}

\begin{figure}[h!]
\centering
\begin{fmffile}{diagram-eye}
\begin{fmfgraph*}(80,50)
    \fmfpen{thin}
    \fmfleft{e1}
    \fmftop{w1}
    \fmfbottom{w2}
    \fmfright{e2}
    \fmflabel{$p_1$}{e1}
    \fmflabel{$-p_1$}{e2}
    \fmf{phantom,tension=5}{w1,v3}
    \fmf{phantom,tension=5}{w2,v4}
    \fmf{plain,tension=1,label=$4$}{v3,v1}
    \fmf{plain,tension=1,label=$3$}{v2,v3}
    \fmf{plain,tension=1,label=$1$}{v4,v1}
    \fmf{plain,tension=1,label=$2$}{v2,v4}
    \fmf{plain,tension=.5,label=$5$,left=.25}{v3,v4}
    \fmf{plain,tension=.5,label=$6$,right=.25}{v3,v4}
    \fmf{plain,tension=5,width=thick}{e1,v1}
    \fmf{plain,tension=5,width=thick}{e2,v2}
\end{fmfgraph*}
\end{fmffile}
\vspace{.1cm}
\caption{Eye diagram}
\label{fig:eyediagram}
\end{figure}
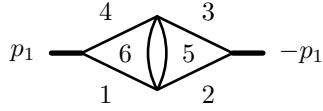
We now consider an example which violates the geometric property, and explain how to deal with it.
Consider the ``eye-diagram'' of Fig. \ref{fig:eyediagram}, in $D=4-2\epsilon$ and choosing $\alpha_6=1$ the corresponding integral is given by
\begin{align}
    I(s;\epsilon) = \Gamma(3\epsilon) \int_{\mathbb{R}^5_+} \frac{d\alpha}{\alpha} \alpha \mathcal{U}^{-2 + 4 \epsilon} \mathcal{F}^{-3\epsilon},
\end{align}
with Symanzik polynomials
\begin{alignat}{2}
&\mathcal{U} = &&\alpha_1 \alpha_2 \alpha_5 + \alpha_1 \alpha_3 \alpha_5 + \alpha_2 \alpha_4 \alpha_5 + \alpha_3 \alpha_4 \alpha_5 + \alpha_1 \alpha_2 \alpha_6 + \alpha_1 \alpha_3 \alpha_6 \notag \\
& &&+ \alpha_2 \alpha_4 \alpha_6 + \alpha_3 \alpha_4 \alpha_6 + \alpha_1 \alpha_5 \alpha_6 + \alpha_2 \alpha_5 \alpha_6 + \alpha_3 \alpha_5 \alpha_6 + \alpha_4 \alpha_5 \alpha_6 \notag\\
&\mathcal{F} = &&s \left(\alpha_1 \alpha_2 \alpha_3 \alpha_5 + \alpha_1 \alpha_2 \alpha_4 \alpha_5 + \alpha_1 \alpha_3 \alpha_4 \alpha_5 + \alpha_2 \alpha_3 \alpha_4 \alpha_5 \right. \notag\\
& &&+\alpha_1 \alpha_2 \alpha_3 \alpha_6 + \alpha_1 \alpha_2 \alpha_4 \alpha_6 + \alpha_1 \alpha_3 \alpha_4 \alpha_6 + \alpha_2 \alpha_3 \alpha_4 \alpha_6 + \alpha_1 \alpha_3 \alpha_5 \alpha_6  \notag\\
& &&+\left. \alpha_2 \alpha_3 \alpha_5 \alpha_6 + \alpha_1 \alpha_4 \alpha_5 \alpha_6 + \alpha_2 \alpha_4 \alpha_5 \alpha_6\right).
\end{alignat}
The integrand has logarithmic divergences on the rays 
\begin{align}
    \rho_1 &= (1,1,1,1,0,0) \\ 
    \rho_2 &= (0,1,1,0,0,0) \\
    \rho_3 &= (1,0,0,1,0,0),
\end{align}
$\rho_1$ is compatible with both $\rho_2$ and $\rho_3$. Due to the relation $\rho_1 = \rho_2 + \rho_3$, the geometric property cannot be satisfied.
To solve the issue, we can use Nilsson-Passare to analytically continue along e.g. $\rho_2$,
this gives
\begin{align}
    \mathcal{I} = \frac{(1-2\epsilon)s^{-3\epsilon}}{\epsilon} \mathcal{I}_1 + \frac{3 s^{-3\epsilon}}{2} \mathcal{I}_2,
\end{align}
the explicit pole in front of $\mathcal{I}_1$ means the corresponding integral has to be computed up to order $\mathcal{O}(\epsilon)$.
The divergent fans are
\begin{align}
    \Sigma^{\rm div}_1 = \{\mathrm{Span}_+ \rho_3 \} 
    \quad \Sigma^{\rm div}_2 = \{\mathrm{Span}_+ \rho_1, \mathrm{Span}_+ \rho_3,  \mathrm{Span}_+ (\rho_1,\rho_3)\},
\end{align}
so that both integrand satisfy the geometric property and can be computed by the subtraction formula.
Putting everything together, we find
\begin{align}
    \Gamma(3\epsilon)^{-1} I(s;\epsilon) = \frac{1}{\epsilon^2} + \frac{7}{\epsilon} + (31 - 6 \zeta(2)) + \mathcal{O}(\epsilon),
\end{align}
which is in agreement with previous results \cite{Chetyrkin:1981qh,Hillman:2023ezp}.

\subsection{A region integral}

We now consider an example which involves a generalization of a Feynman integral,
\begin{align}
    I(s; \epsilon, \mu) = s^{-(1+\epsilon+\mu)} \int_{\mathbb{R}^3_\ge} \frac{d\alpha}{\alpha} \alpha^{(1+\mu,1,1-\mu)} (1 + \alpha_1 + \alpha_2)^{2\epsilon} (\alpha_2 + t \alpha_1 \alpha_3)^{-(2+\epsilon)}.
    \label{eq:regionintegral}
\end{align}
Such integrals appearing in the context of the \emph{Method of Regions}, a technique useful to study Feynman integrals around kinematical limits. The particular example \eqref{eq:regionintegral} corresponds to one of the two \emph{soft regions} contributing to the Regge limit, $s \to 0$, of a massless box in $D = 6-2\epsilon$.
An interesting feature of this example is the use of an additional exponent variable, the analytic regulator $\mu$. The poles in $\mu$ and $\epsilon$ combine with the prefactor $s^{-(1+\epsilon\mu)}$ to give rise to logarithms in $s$ in the final result.
In physical applications, such as the study of effective field theories, it is useful to be able to disentangle the logarithms originating from poles in the different regulators.

Let us assume that $-\epsilon \ge \mu \ge 0$, and let us study the expansion in $\mu$.
The divergent fan $\Sigma^{\rm div}$ consists of a single ray $\rho = (-1,0,1)$. We have,
\begin{align}
    w_\rho = (1,0,0), \quad \rm{Vol}(\rho) = \frac{1}{2 \mu}, \quad \eta_\rho = (0,0), \quad V_\rho = \{\alpha_1 = 1\}
\end{align}
and 
\begin{align}
    \mathcal{I}|_\rho = \alpha_1^\mu \alpha_3^{-\mu} (1+\alpha_2)^{2\epsilon}(\alpha_2 + t \alpha_1 \alpha_3)^{-(2+\epsilon)}
\end{align}
Accordingly the subtraction formula gives
\begin{align}
    s^{(1+\epsilon+\mu)}  I(s; \epsilon, \mu) = 
    \int_{\mathbb{R}^3} \frac{d\alpha}{\alpha} \left(\mathcal{I}-v_\rho^{1+2\mu}  \mathcal{I}|_{\rho}\right)
    + \frac{1}{2\mu} \int_{V_\rho} \frac{d\alpha}{\alpha} \mathcal{I}|_{\rho}.
\end{align}
The integrands can now be expanded in $\mu$. 
The expansion in $\epsilon$ of the second integrand can be performed in the same way. On the other hand, the technique is not immediately applicable to the first combination of integrands. We will leave to a future work the study of this interesting problem.

Finally, the simultaneous expansion in $\mu$ and $\epsilon$ can be performed as usual.
Considering the line $\epsilon = -2 \mu$, there are three divergent rays. The geometric property is not satisfied, but this can be easily addressed by Nilsson-Passare analytical continuation along the ray $\rho = (-1,0,1)$.
The result is
\begin{align}
    s^{(1+\epsilon+\mu)}  I(s; -2\mu, \mu) = 2
    \int_{\mathbb{R}^3_\ge} \frac{d\alpha}{\alpha} \alpha^{(2+\mu,1,1-\mu)} (1 + \alpha_1 + \alpha_2)^{-1-4\mu} (\alpha_2 + t \alpha_1 \alpha_3)^{-(2+\mu)},
\end{align}
which now satisfies the geometric property. Note that no explicit pole in $\mu$ has been produced in the process. The subtraction formula gives
\begin{align}
    \frac{t}{2} s^{(1+\epsilon+\mu)}  I(s; -2\mu, \mu) = \frac{1}{\mu^2} + \frac{1}{\mu} \left(\log(t)+2\right)  + 2 \left(2-2\zeta(2)+\log(t)\right) + \mathcal{O}(\mu).
\end{align}

\subsection{An angular integral}

We provide another example beyond the case of Feynman integrals, a class of phase-space integrals known as \emph{angular integrals}.
In \cite{smirnov2024expansion} it was introduced a parametric formulation which allows to apply our master formula to this case as well.
In the language of \cite{smirnov2024expansion}, we consider the four denominators, one-mass integral, given by
\begin{align}
    I({\bf v};\epsilon) = \mathcal{N} \int_{\mathbb{R}^4_+} \frac{dt}{t} t^{(1,1,1,5/2)} (1+t_2)(1+t_3)^2(1+t_4)^{\epsilon+1/2} \mathcal{Q}^{-2},
\end{align}
where the normalization factor is
\begin{align}
    \mathcal{N} = \frac{2\pi}{1-2\epsilon} \frac{\Gamma(4) \Gamma(\frac{3}{2}-\epsilon)}{\Gamma(\frac{5}{2})\Gamma(-\epsilon-1)}  = 8 \pi \epsilon + \mathcal{O}(\epsilon^2)
\end{align}
and the $\mathcal{Q}$ is a quadratic polynomial,
\begin{align*}
    \mathcal{Q}=&(t_1+1)^2 (t_2+1)^2
   (t_3+1)^2+\\
   &t_4 \left[t_3 v_{34}+2 t_1 (t_2+1) (t_3+1) (t_2 (t_3+1)
   v_{12}+t_3 v_{13}+v_{14}) +2 t_2
   (t_3+1) (t_3 v_{23}+v_{24})+v_{44}\right]\;,
\end{align*}
Note that due to the prefactor, in order to compute the integral up to the finite part we need to compute $I({\bf v};\epsilon) \mathcal{N}^{-1}$ only up to the first pole in $\epsilon$. This is a situation where the subtraction method is intrinsically superior to the Nilsson-Passare analytic continuation.

The integrand has a power-divergence along the ray $\rho = (0,0,0,1)$, which requires analytic continuation in that direction to reduce log-divergent integrands. This produces two new integrands that satisfy the geometric property, hence we can apply our subtraction formula to them.
Although the resulting integrands are not linearly reducible, they can still be integrated with \verb|HyperInt| by rationalizing the last integration.
In particular, the final result contains an interesting square root,
\begin{align*}
    r_{\pm} = -\frac{v_{14} v_{23} + v_{13} v_{24} - v_{12} v_{34} \pm \sqrt{( v_{14} v_{23} + v_{13} v_{24} - 2 v_{12} v_{34})^2 + 2 v_{13} v_{23} (v_{12} v_{44}-2 v_{14} v_{24} )
    }}{2 v_{13} v_{23} }.
\end{align*}
The result is given by
\begin{align}
    I({\bf v}; \epsilon) = \mathcal{N} \left(\frac{1}{\epsilon^2} J^{(-2)} + \frac{1}{\epsilon} J^{(-1)} \right) + \mathcal{O}(\epsilon),
\end{align}
where
\begin{align}
    J^{(-2)}= -\frac{v_{12} v_{14} v_{24} + v_{13} v_{14} v_{34} + v_{23} v_{24} v_{34}}{
 8 v_{12} v_{13} v_{14} v_{23} v_{24} v_{34}},
\end{align}
the expression for $J^{(-1)}$ involves polylogarithmic functions of weight two, and is reported in the ancillary files\footnote{A mistake was done in the first version of this manuscript, and the reported result was wrong. We would like to thank V.Smirnov and F.Wunder for noticing the mistake}. The analytic expression for this angular integral was previously unknown.

\subsection{Examples with $n+L \ge 7$}

We conclude with an exhibition of less trivial diagrams satisfying the geometric property, showing that it is not a consequence of the simplicity of the previous examples.
At the time of writing, a rough characterization of the state of the art in the computation of Feynman integrals is given by the quantity $n+L$, where $n$ is the number of external particles and $L$ the number of loops. Diagrams below the critical line $n+L=7$ are mostly understood, with some exceptions due to extremely complicated topologies and number of kinematical variables.
In Fig. \ref{fig:various examples} we show some examples beyond this critical line, all of which satisfy the geometrical property.
\begin{figure}[t]
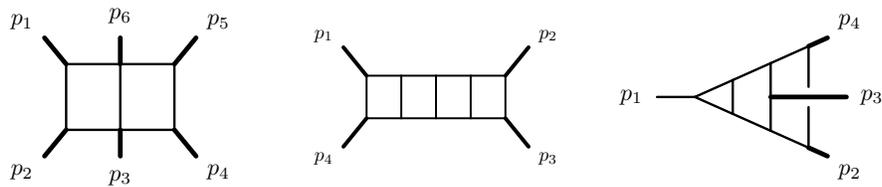

    \vspace*{3mm}
    \centering
    \begin{subfigure}[c]{.27\textwidth}
        \centering
        \resizebox{.7\textwidth}{!}{\input{doublehexabox}}
    \end{subfigure}%
    \begin{subfigure}[c]{.27\textwidth}
        \centering
        \resizebox{.7\textwidth}{!}{\input{triplebox}}
    \end{subfigure}%
    \begin{subfigure}[c]{.27\textwidth}
        \centering
        \resizebox{.7\textwidth}{!}{\input{4pt3loop_nonplanar_1}}
    \end{subfigure}%
    \vspace*{1.5mm}
    \caption{A collection of Feynman integrals which can be treated by our subtraction scheme. At least one of the solid lines must carry a momentum with $p_i^2 \ne 0$.}
    \label{fig:various examples}
\end{figure}

\section{Conclusion} 
\label{sec:outlook}

In this paper we have applied tropical geometry to the study of local subtraction schemes for Euler integrals.
Our first main result is a criterion for a combination of Euler integrands to be locally finite, expressed by Th. \ref{th:localfinite3} and Th. \ref{th:localfinitesuf}.
We have applied this theorem, and the tropical wisdom that goes into its proof, to construct a simple-minded subtraction scheme for Euler integrals which satisfy the geometrical property \ref{def:geometricproperty}.
The scheme can be summarized by the subtraction formula described in Th. \ref{th:subformula}, which is the central result of the paper.
The Nilsson-Passare analytical continuation allows to algorithmically express any Euler integral as a combination of integrals satisfying the assumption of Th. \ref{th:subformula}, so our results apply to arbitrary Euler integrals.
A private implementation of this method is available upon request to the author, and it will be made public in the future.

There are several natural directions for further investigation.

{\bf Relative signs.} What if the coefficients of ${\bf s}$ have relative signs?
These Euler integrands are the \emph{b\^ete noire} of tropicalization, because they may develop singular behaviours around loci in the \emph{interior} of the domain of integration, to which the blow-ups given by the monomial changes \eqref{eq:monomialChange} are oblivious.
This may result in the appearance of additional poles in $\epsilon$, which in particular would invalidate the subtraction scheme constructed in Section \ref{sec:scheme}.
While we are not aware of examples of this phenomenon appearing in Feynman integrals, we are also not aware of a rigorous no-go theorem ruling out this possibility.
It would be interesting to investigate this aspect by applying recent developments in semi-algebraic geometry \cite{telek2023geometry} to the Landau analysis.

{\bf Geometrical Property}.
The main limitation of our subtraction scheme is that it applies only to Euler integrands satisfying the geometric property.

One way to deal with this issue is to apply Nilsson-Passare analytical continuation along enough divergent rays, until the integral is expressed in terms of a combination of integrands satisfying the geometric property. 
Although practical, this approach goes against the original motivation of the subtraction.
A better solution would be instead to find suitable counterterms to deal with these cases, by considering  more general rational functions for the factors $v_\rho$. 

A considerably different approach is to construct the counterterms in terms of exact differential forms, which makes it easier to integrate them.
This is the approach followed in \cite{Brown:2019wna} for string amplitudes and in \cite{Hillman:2023ezp} for UV divergent Feynman integrals.
The idea is to find \emph{u-variables}, rational functions $u_\rho$ taking values in $[0,1]$, with suitable asymptotic behaviors along the divergent rays. Then one construct counterterms using logarithmic differential forms $d\log(u)$. It would be interesting to connect these approaches with the formalism established here, in particular with the criterion for local finiteness. This would require to study the tropicalization of the \emph{Toric Jacobian}, $\mathcal{J}$, defined by \cite{cattani1997residues}
\begin{align}
    \mathcal{J} \frac{d\alpha}{\alpha} \coloneqq \frac{du_1}{u_1} \wedge \dots \wedge \frac{du_n}{u_n} .
\end{align}

{\bf Multiple expansions}
As mentioned in one of the examples, it is sometimes interesting to look at the multiple expansion of an Euler integral in several exponent variables.
In general, one does not expect the expansions to commute, and different orders may reveal different physical aspects of the problem under consideration.
While the main subtraction formula is not immediately applicable to study this problem, we expect it to be possible to suitably modify it with the aid of the local finiteness criterion explained in this paper.

{\bf Finite Integrals.}
In various situations in physics one often encounters combinations of Feynman integrals - and generalizations thereof - which are finite but are not presented in a locally finite presentation from the get-go.
A famous and simple example is the following combination of Feynman integrals \cite{Henn_2021},
\vspace{2mm}
\begin{alignat*}{1}
I(s,t;\epsilon) = \quad\quad {\vcenter{\hbox{\resizebox{.081\textwidth}{!}{\fontsize{80}{90} \begin{fmffile}{diagram-box}
\begin{fmfgraph*}(80,50)
    \fmfpen{thin}
    \fmfleft{e2,e1}
    \fmfright{e3,e4}
    \fmflabel{$p_1$}{e1}
    \fmflabel{$p_2$}{e2}
    \fmflabel{$p_3$}{e3}
    \fmflabel{$p_4$}{e4}
    \fmf{plain,tension=.4, label=$1$}{v1,v2}
    \fmf{plain,tension=.4, label=$2$}{v2,v3}
    \fmf{plain,tension=.4, label=$3$}{v3,v4}
    \fmf{plain,tension=.4, label=$4$}{v4,v1}
    \fmf{plain}{e1,v1}
    \fmf{plain}{e2,v2}
    \fmf{plain}{e3,v3}
    \fmf{plain}{e4,v4}
\end{fmfgraph*}
\end{fmffile}} }  }} 
 -\frac{1}{t}\times \ {\vcenter{\hbox{\resizebox{.081\textwidth}{!}{\fontsize{80}{90}\begin{fmffile}{diagram-pbox1}
\begin{fmfgraph*}(80,50)
    \fmfpen{thin}
    \fmfleft{e2,e1}
    \fmfright{e3,e4}
    \fmflabel{$p_1$}{e1}
    \fmflabel{$p_2$}{e2}
    \fmflabel{$p_3$}{e3}
    \fmflabel{$p_4$}{e4}
    \fmf{plain,tension=100}{v1,v2}
    \fmf{plain,tension=.2, label=$2$}{v3,v2}
    \fmf{plain,tension=.2, label=$3$}{v3,v4}
    \fmf{plain,tension=.2, label=$4$}{v1,v4}
    \fmf{plain}{e1,v1}
    \fmf{plain}{e2,v2}
    \fmf{plain}{e3,v3}
    \fmf{plain}{e4,v4}
\end{fmfgraph*}
\end{fmffile}} }  }}
 -\frac{1}{t}\times \ {\vcenter{\hbox{\resizebox{.081\textwidth}{!}{\fontsize{80}{90}\begin{fmffile}{diagram-pbox3}
\begin{fmfgraph*}(80,50)
    \fmfpen{thin}
    \fmfleft{e2,e1}
    \fmfright{e3,e4}
    \fmflabel{$p_1$}{e1}
    \fmflabel{$p_2$}{e2}
    \fmflabel{$p_3$}{e3}
    \fmflabel{$p_4$}{e4}
    \fmf{plain,tension=.2, label=$1$}{v1,v2}
    \fmf{plain,tension=.2, label=$2$}{v3,v2}
    \fmf{plain,tension=100}{v3,v4}
    \fmf{plain,tension=.2, label=$4$}{v4,v1}
    \fmf{plain}{e1,v1}
    \fmf{plain}{e2,v2}
    \fmf{plain}{e3,v3}
    \fmf{plain}{e4,v4}
\end{fmfgraph*}
\end{fmffile}} }  }} 
 -\frac{1}{s}\times \ {\vcenter{\hbox{\resizebox{.081\textwidth}{!}{\fontsize{80}{90}\begin{fmffile}{diagram-pbox2}
\begin{fmfgraph*}(80,50)
    \fmfpen{thin}
    \fmfleft{e2,e1}
    \fmfright{e3,e4}
    \fmflabel{$p_1$}{e1}
    \fmflabel{$p_2$}{e2}
    \fmflabel{$p_3$}{e3}
    \fmflabel{$p_4$}{e4}
    \fmf{plain,tension=.2, label=$1$}{v2,v1}
    \fmf{plain,tension=100}{v2,v3}
    \fmf{plain,tension=.2, label=$3$}{v4,v3}
    \fmf{plain,tension=.2, label=$4$}{v4,v1}
    \fmf{plain}{e1,v1}
    \fmf{plain}{e2,v2}
    \fmf{plain}{e3,v3}
    \fmf{plain}{e4,v4}
\end{fmfgraph*}
\end{fmffile}} }  }}
 -\frac{1}{s}\times \ {\vcenter{\hbox{\resizebox{.081\textwidth}{!}{\fontsize{80}{90}\begin{fmffile}{diagram-pbox4}
\begin{fmfgraph*}(80,50)
    \fmfpen{thin}
    \fmfleft{e2,e1}
    \fmfright{e3,e4}
    \fmflabel{$p_1$}{e1}
    \fmflabel{$p_2$}{e2}
    \fmflabel{$p_3$}{e3}
    \fmflabel{$p_4$}{e4}
    \fmf{plain,tension=.2, label=$1$, label.side=right}{v1,v2}
    \fmf{plain,tension=.2, label=$2$}{v2,v3}
    \fmf{plain,tension=.2, label=$3$}{v4,v3}
    \fmf{plain,tension=100}{v4,v1}
    \fmf{plain}{e1,v1}
    \fmf{plain}{e2,v2}
    \fmf{plain}{e3,v3}
    \fmf{plain}{e4,v4}
\end{fmfgraph*}
\end{fmffile}} }  }}.
\end{alignat*}
\vspace{2mm}

When trying to find a locally finite presentation for $I(s,t;\epsilon)$, one faces an obstacles in that the Feynman parametric integrals associated to the diagrams live in different spaces.
An obvious solution is to consider the following ``global'' Schwinger parametrization,
\begin{align}
    I(s,t; {\bf \nu}, \epsilon) &= \frac{1}{(i \pi)^{D/2}}\int \frac{d^D \ell}{[\ell^2]^{\nu_1} [(\ell+p_2)^2]^{\nu_2}  [(\ell+p_2+p_3)^2]^{\nu_3}  [(\ell-p_1)^2]^{\nu_4}}\\
    &=\Gamma\left(\sum_{i=1}^4 \nu_i - D/2\right) \int_{\mathbb{P}^3_{\ge0}} \frac{1}{\mathrm{GL}(1)} \frac{d\alpha}{\alpha} \alpha^{\bf \nu} \mathcal{U}^{\left(\sum_{i=1}^4 \nu_i-D\right)} \mathcal{F}^{-\left(\sum_{i=1}^4 \nu_i-D/2\right)},
    \label{eq:family}
\end{align}
where ${\bf \nu} = (\nu_1, \dots, \nu_4)$.
The Symanzik polynomials are
\begin{alignat}{2}
&\mathcal{U} = && \alpha_1 + \alpha_2 + \alpha_3 + \alpha_4  \notag \\
&\mathcal{F} = && s \alpha_1 \alpha_3 + t \alpha_2 \alpha_4 \notag
\end{alignat}
The massless box diagram $I_{\rm box}$ appearing as the first term in $I(s,t;\epsilon)$ corresponds to ${\bf \nu} = (1, \dots, 1)$, while the triangle obtained by pinching the $i$-th edge corresponds to ${\bf \nu}_j = 1-\delta_{i,j} \mu$, for small $\mu$:
\begin{align}
I_{t_i} = I(s,t; \nu_j = 1 - \delta_{i,j}\mu, \epsilon) + \mathcal{O}(\mu).
\end{align}
Plugging the parametrization \eqref{eq:family} into $I(s,t;\epsilon)$, however, \emph{does not} yield a locally finite presentation. We can see this using Th. \ref{th:localfinitesuf}: each triangle integrand $\mathcal{I}_{t_i}$ has a log-divergence along a ray $\rho_i$ which, being unique to that particular integrand, cannot be cancelled. 
Informed by this, we apply Nilsson-Passare analytically continue $\mathcal{I}_{t_i}$ along $\rho_i$, after which the parametrization collapses to
\begin{align}
    I(s,t;\epsilon) &= 2(-1+2\epsilon) \Gamma(1+\epsilon)\int_{\mathbb{R}^3_{\ge0}} \frac{d\alpha_1 d\alpha_2 d\alpha_3}{(1+\alpha_1+\alpha_2+\alpha_3)^{2-\epsilon}(s \alpha_2 + t \alpha_1 \alpha_3)^{1+\epsilon}} + \mathcal{O}(\mu) \\
    &= -\log\left(\frac{s}{t}\right)^2-\pi^2 + \mathcal{O}(\mu,\epsilon),
\end{align}
which is locally finite. 

It would be interesting to employ the same ideas to find finite combinations of Feynman integrals - such as those appearing in \cite{Henn_2023,Arkani_Hamed_2022,Arkani_Hamed_2012} - and put them in a locally finite form.
See \cite{gambuti2023finite} for an alternative approach based on loop momentum space.

{\bf Linear Reducibility.}
Linear reducibility is a property satisfied by many important Euler integrands which allows, when the integrals are convergent, to integrate them in terms of Goncharov polylogarithms \cite{Panzer_2015, brown2010periods, Brown:2008um}. Geometrically, it requires the iterated projections of the \emph{Landau variety} to be contained in linear varieties.

Starting from a linearly reducible integrand $\mathcal{I}$ which satisfies the geometric property, there is no guarantee that the locally finite integrands produced by the subtraction scheme will also be linearly reducible.
The property may be broken by the additional components of the Landau variety introduced by the factors $v_\rho$, as well as by the initial forms $\mathcal{I}|_\sigma$.
Therefore, requiring linear reducibility to be preserved poses additional constraints on the structure of the allowed counterterms. This is essentially an algebraic geometrical problem, which would be interesting to spell out in order to characterize the integrals for which linear reducibility can be preserved, as well as to design algorithms to construct the required counterterms directly.

{\bf Numerical Integration}
In most cases, the subtraction formula involves integrals that cannot be presented in terms of known special functions, this motivates the search for numerical methods for their evaluation.

A natural concern is that the cancellations of divergences inherent to any subtraction scheme may be difficult to detect numerically.
However, an advantage of our scheme is that the regions in the integration space where the cancellations take place are completely understood. They correspond to the cones of the common refiment of the integrands, so that the cancellations can be accurately sampled after performing the monomial changes used in the construction of the scheme. One can even see the cancellations analytically, in each cone, by setting $\epsilon=0$ and factoring the combination of integrands.

A recent breakthrough in the context of Feynman integrals, \emph{tropical sampling}, seems a particularly promising approach to apply to our formula \cite{Borinsky_2023,feyntrop}. It would be interesting to test these ideas on the tropicalization of the local integrands produced by the subtraction formula.

{\bf Beyond Feynman Integrals}
The subtraction scheme described in this paper is applicable to any Euler integrand satisfying the geometric property, as it was shown in the provided examples.
A particularly satisfying aspect of this framework is that it illuminates the combinatorics of the cancellations. Traditional challenges of subtraction schemes, such as treating double-counting of divergences, are neatly dealt by the face lattice of the Newton polytope.
Based on this, we advocate for the search of parametrizations of integrals appearing in various physical applications, such as cosmology and phase-space integrations, in order to express these quantities in terms of Euler integrals.

{\bf Subtraction formulae for Curve Integrals}
Part of the motivation behind the research presented in this manuscript is the recent introduction of the \emph{curve integral} presentation for scattering amplitudes \cite{arkanihamed2023loop,curve2}. In this formalism, scattering amplitudes are described as a single integral of an integrand which can be efficiently evaluated at every integration point. Surprisingly, the enumeration of Feynman diagrams, which is a major bottleneck in the computation of scattering amplitudes, is completely avoided.
While the formalism was originally developed for $\rm{Tr\ } \phi^3$ theory, it was then extended to more general scalar theories, and there is growing evidence that efficient formulations should exist for theories describing the real world \cite{Arkani-Hamed:2024vna,Arkani-Hamed:2023jry, Arkani-Hamed:2024nhp,Laddha:2024qtn, De:2024wsy}.

An exciting prospect is that the curve integral formalism may result in a powerful numerical method to compute scattering amplitudes in phenomenologically relevant theories. In order for this prospect to concretize, however, several challenges have to be faced, among which is the treatment of UV and IR divergences.
A natural approach would the construction of a subtraction scheme for the curve integrand. This would allow, after the local removal of UV and IR divergences, to evaluate numerically finite observables without generating Feynman diagrams.

\section{Acknowledgments}

GS would like to thank Aaron Hillman and Nima Arkani-Hamed for various discussions around the topic of subtraction schemes, Johannes Henn for suggesting physical problems that have helped in orienting the investigation of this topic, as well as for valuable comments on the first draft of the manuscript, Giulio Ruzza for discussions about the commutativity of series expansions and integrals and Veronika Salvatori for saying yes.
This project was funded by the European Union (ERC, UNIVERSE PLUS, 101118787). Views and opinions expressed are however those of the author(s) only and do not necessarily reflect those of the European Union or the European Research Council Executive Agency. Neither the European Union nor the granting authority can be held responsible for them.

\newpage
\newpage

\bibliographystyle{unsrt}
\bibliography{refs}

\end{document}